\documentclass[sigconf, nonacm]{acmart}

\newcommand\vldbdoi{XX.XX/XXX.XX}

\newcommand\vldbpagestyle{plain}

\usepackage{graphicx}
\usepackage{booktabs}
\usepackage[skip=5pt]{caption}
\usepackage{subcaption}
\usepackage{multirow}
\usepackage{weiwAlgorithm}
\usepackage{amsmath}
\usepackage{newtxmath}
\usepackage{diagbox}
\usepackage{url}
\usepackage{dsfont}
\usepackage{enumitem}
\usepackage{caption}
\usepackage{balance}
\usepackage{tabularx}

\newcommand{\nonl}{\renewcommand{\nl}{\let\nl\oldnl}}

\newtheorem{example}{Example}[section]
\newtheorem{theorem}{Theorem}[section]

\newtheorem{lemma}{Lemma}[section]

\newtheorem{definition}{Definition}[section]

\newcommand{\myparagraph}[1]{\vspace{1mm} \noindent \textbf{#1}.}

\newcommand{\exppara}[1]{\noindent \textbf{#1}.}

\newcommand{\ie}{{i.e.,}\xspace}
\newcommand{\eg}{{e.g.,}\xspace}
\newcommand{\etal}{et al.\xspace}

\newcommand{\mrg}{\textsf{MFG}\xspace}
\newcommand{\mrgs}{\textsf{MFG}s\xspace}
\newcommand{\rg}{$\lambda$-frequency group\xspace}

\newcommand{\sstrneighbor}{{s-neighbor}\xspace}
\newcommand{\sstrneighbors}{{s-neighbors}\xspace}
\newcommand{\strneighborset}[1]{{N(#1, {G})}\xspace}
\newcommand{\strneighbordeg}[1]{{d(#1, {G})}\xspace}
\newcommand{\fstrdegree}{{structural degree}\xspace}
\newcommand{\sstrdegree}{{s-degree}\xspace}

\newcommand{\smomeneighbor}{{m-neighbor}\xspace}
\newcommand{\smomeneighbors}{{m-neighbors}\xspace}
\newcommand{\fmomedegree}{{momentary degree}\xspace}
\newcommand{\smomedegree}{{m-degree}\xspace}
\newcommand{\momeneighborset}[2]{{\Gamma(#1, #2)}\xspace}
\newcommand{\momeneighbordeg}[2]{{\delta(#1, #2)}\xspace}
\newcommand{\frequent}{{frequency}\xspace}

\newcommand{\nflag}{notRepeat\xspace}

\newcommand{\bkalg}{\textsf{BK-ALG}\xspace}
\newcommand{\bkalgp}{\textsf{BK-ALG$+$}\xspace}
\newcommand{\filtervr}{\textsf{FilterV-}\xspace} 
\newcommand{\filterv}{\textsf{FilterV}\xspace}
\newcommand{\filtervfr}{\textsf{FilterV-FR}\xspace} 
\newcommand{\filtervvm}{\textsf{FilterV-VM}\xspace} 
\newcommand{\filtervcm}{\textsf{FilterV-CM}\xspace}
\newcommand{\vfree}{\textsf{VFree}\xspace}
\newcommand{\vfreer}{\textsf{VFree-}\xspace} 
\newcommand{\vfreecm}{\textsf{VFree-CM}\xspace}

\newcommand{\twoarr}{array\xspace}

\newcommand{\abcore}[1]{\textsf{Core}(#1)\xspace}

\begin{document}
\title{Efficient Maximal Frequent Group Enumeration in Temporal Bipartite Graphs}

\author{Yanping Wu}
\affiliation{%
  \institution{University of Technology Sydney}
}
\email{yanping.wu@student.uts.edu.au}

\author{Renjie Sun}
\affiliation{%
  \institution{East China Normal University}
}
\email{renjie.sun@stu.ecnu.edu.cn}

\author{Xiaoyang Wang}
\affiliation{%
  \institution{The University of New South Wales}
}
\email{xiaoyang.wang1@unsw.edu.au}

\author{Dong Wen}
\affiliation{%
  \institution{The University of New South Wales}
}
\email{dong.wen@unsw.edu.au}

\author{Ying Zhang}
\affiliation{%
  \institution{University of Technology Sydney}
}
\email{ying.zhang@uts.edu.au}

\author{Lu Qin}
\affiliation{%
  \institution{University of Technology Sydney}
}
\email{lu.qin@uts.edu.au}


\author{Xuemin Lin}
\affiliation{%
  \institution{Shanghai Jiaotong University}
}
\email{xuemin.lin@sjtu.edu.cn}


\begin{abstract}
Cohesive subgraph mining is a fundamental problem in bipartite graph analysis. In reality, relationships between two types of entities often occur at some specific timestamps, which can be modeled as a temporal bipartite graph. However, the temporal information is widely neglected by previous studies. Moreover, directly extending the existing models may fail to find some critical groups in temporal bipartite graphs, which appear in a unilateral (i.e., one-layer) form. To fill the gap, in this paper, we propose a novel model, called maximal $\lambda$-\frequent group (\mrg). Given a temporal bipartite graph $\mathcal{G}=(U,V,\mathcal{E})$, a vertex set $V_S \subseteq V$ is an \mrg if $i)$ there are no less than $\lambda$ timestamps, at each of which $V_S$ can form a $(\tau_U,\tau_V)$-biclique with some vertices in $U$ at the corresponding snapshot, and $ii)$ it is maximal. To solve the problem, a filter-and-verification (\filterv) method is proposed based on the Bron\text{-}Kerbosch framework, incorporating novel filtering techniques to reduce the search space and \twoarr-based strategy to accelerate the frequency and maximality verification. Nevertheless, the cost of frequency verification in each valid candidate set computation and maximality check could limit the scalability of \filterv to larger graphs. Therefore, we further develop a novel verification-free (\vfree) approach by leveraging the advanced dynamic counting structure proposed. Theoretically, we prove that \vfree can reduce the cost of each valid candidate set computation in \filterv by a factor of $\mathcal{O}(|V|)$. Furthermore, \vfree can avoid the explicit maximality verification because of the developed search paradigm. Finally, comprehensive experiments on 15 real-world graphs are conducted to demonstrate the efficiency and effectiveness of the proposed techniques and model. 
\end{abstract}

\maketitle

\pagestyle{\vldbpagestyle}


\section{Introduction}
\label{sec:intro}

Bipartite graphs are widely used to model the complex relationships between two types of entities, e.g.,
author-paper networks~\cite{wang2022towards,wang2020efficient}, 
customer-product networks~\cite{lyu2020maximum,yang2021efficient,beutel2013copycatch,kim2022abc,zhou2021butterfly}, 
and patient-disease networks~\cite{aziz2021multimorbidity,zhao2019integrating,krishnagopal2020identifying}.
As a fundamental problem in graph analysis, cohesive subgraph mining is broadly investigated. 
To analyze the properties of bipartite graphs, many cohesive subgraph models have been proposed, such as $(\alpha,\beta)$-core~\cite{DBLP:conf/www/LiuYLQZZ19}, bitruss~\cite{wang2022towards} and biclique~\cite{lyu2020maximum,sun2022maximal,yao2022identifying,DBLP:journals/tkde/SunWWCZL24}.
Among these models, biclique, which requires every pair of vertices from different vertex sets to be mutually connected, has gained widespread popularity due to its unique features and diverse applications.
However, existing models on bipartite graphs primarily focus on static graphs, disregarding the temporal aspect of relationships in real-world applications.
For instance, Figure~\ref{fig:intro} shows a customer-product network. The timestamps on an edge between two vertices indicate when a customer purchased the corresponding product.
In a patient-disease network, patients and diseases can be represented by two disjoint vertex sets, and the timestamps on an edge represent when the patient suffered from the disease. 
The above cases can be modeled as a temporal bipartite graph $\mathcal{G}=(U,V,\mathcal{E})$, where each edge $e\in \mathcal{E}$ can be represented as a tuple $(u,v,t)$, indicating that the interaction between vertex $u \in U$ and vertex $v\in V$ occurs at timestamp $t$ (e.g., \cite{aziz2021multimorbidity,chen2021efficiently}).
The bipartite graph with all the edges at timestamp $t$ is called a snapshot of $\mathcal{G}$, denoted by $G_t$.

Analyzing the properties of temporal bipartite graphs is essential to reveal more sophisticated semantics.
In the literature, many subgraph models are defined for unipartite scenarios.
For instance, Li~\etal~\cite{li2018persistent} propose a $k$-core based model to capture the persistence of a community within the time interval.
Qin~\etal~\cite{qin2019mining} design a clique based model to find the subgraphs periodically occurring in the temporal graphs.
Qin~\etal~\cite{qin2022mining} 
propose the $(l,\delta)$-maximal bursting core, where each vertex has an average degree no less than $\delta$ during a period of length no less than $l$.
Although temporal subgraph search and enumeration problems have been extensively studied on temporal unipartite graphs, the temporal bipartite graph case is still under-explored.
Due to the involvement of two distinct types of entities,
existing models on temporal unipartite graphs are not suitable for capturing important patterns in temporal bipartite graphs.
Only a few recent works consider temporal bipartite graphs, such as $(\alpha,\beta)$-core based persistent community search~\cite{li2023persistent} and temporal butterfly counting and enumeration~\cite{cai2023efficient}.
Unfortunately, the aforementioned temporal models are mainly established on interval-based or periodic-based constraints, which fail to capture the real scenarios occurring irregularly at non-consecutive timestamps.
Additionally, there is no existing work considering the \textit{unilateral} (i.e., one-layer) frequency of the occurred group in temporal bipartite graphs, which is an important factor in identifying practical groups.
For example, in a temporal customer-product bipartite graph, a group of users who frequently act together (e.g., purchase the same items at different timestamps) may show strong connections. 
Existing works about temporal subgraph mining are inadequate for the above scenario, highlighting the need to define the exclusive model tailored for temporal bipartite graphs.


To fill the gap, in this paper, we propose a novel model, called \textit{\underline{M}aximal $\lambda$-\underline{F}requency \underline{G}roup} (\mrg), to characterize the unilateral patterns in temporal bipartite graphs.  
In our model, we leverage \textit{biclique} to measure the cohesiveness of bipartite subgraphs, due to its diverse applications such as social recommendation~\cite{liu2006efficient} and anomaly detection~\cite{lyu2020maximum}.
Specifically, given two size constraints $\tau_U$, $\tau_V$ and a frequency constraint $\lambda$, a vertex set $V_S \subseteq V$ in a temporal bipartite graph $\mathcal{G}=(U,V,\mathcal{E})$ is an \mrg if 
$i)$ there are no less than~$\lambda$ timestamps, at each of which $V_S$ can form a $(\tau_U,\tau_V)$-biclique with some vertices in $U$ in the corresponding snapshot graph, 
and $ii)$ it is maximal. 
A $(\tau_U, \tau_V)$-biclique $S=(U_S,$ $V_S,E_S)$ is a biclique with $|U_S|\geq\tau_U$ and $|V_S|\geq\tau_V$.
Note that, an \mrg $V_S$ is unilateral, i.e., only consists of vertices from $V$.
In addition, the vertices in $U$ that form $(\tau_U,\tau_V)$-bicliques with $V_S$ can be different for different timestamps.

\begin{figure}[t] 
	\centering
	\includegraphics[width=0.9\linewidth]{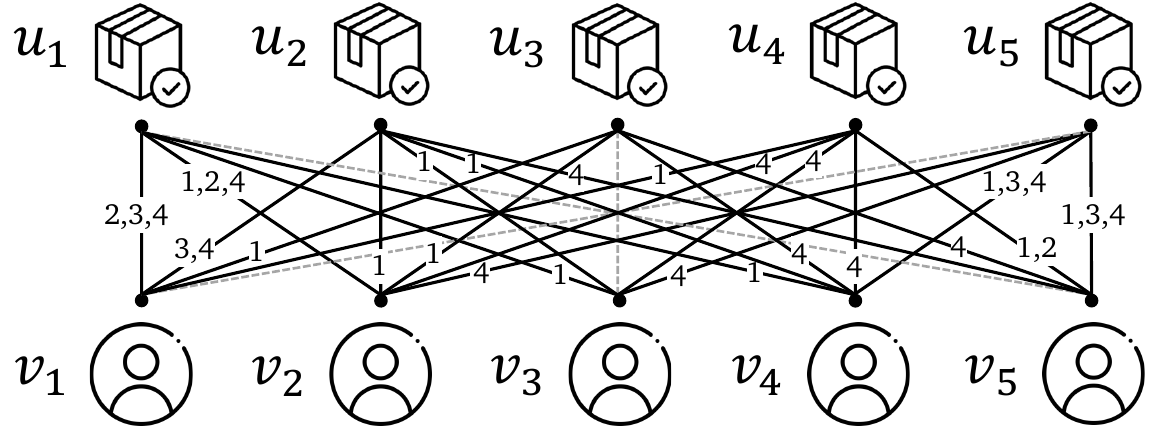}
        \caption{Customer-product temporal bipartite graph (dotted lines denote the edges at $t=5$ for presentation simplicity)}
	\label{fig:intro}
\end{figure}

\begin{example}\label{example:intro}
    Reconsider the temporal customer-product bipartite graph in Figure~\ref{fig:intro} with $\tau_U=2$, $\tau_V=2$ and $\lambda=2$. 
    Suppose a company plans to promote a new product to a group of customers with similar interests. 
    Directly extending the static graph model (i.e., biclique) to temporal bipartite graphs may fail to retrieve useful patterns.
    For instance, if we treat it as a static graph, i.e., ignore the timestamp information, the whole graph itself is a biclique. All five customers will be grouped together and treated equally since they all purchased all the products.
    If we consider it as a temporal bipartite graph, we still cannot find practical results by applying the frequent $(\tau_U, \tau_V)$-biclique model, which is the $(\tau_U, \tau_V)$-biclique occurring in at least $\lambda$ snapshots. 
    For the \mrg model, $\{v_2,v_3,v_4\}$ is the returned result, since these customers frequently act together, i.e., it exists in biclique $\{u_1, u_2,{v_2},$ ${v_3},{v_4}\}$ at timestamp $t=1$ and in biclique $\{u_4,u_5,{v_2},{v_3},{v_4}\}$ at timestamp $t=4$.
    As discussed before, customers within a group, who frequently act together, are more likely to share similar preferences or behavioral patterns.
    Identifying these customer-specified communities is crucial to improve the performance of downstream tasks such as product recommendation and enhance customer engagement.
\end{example}    

    Besides biclique, many other cohesive subgraph models are proposed in the literature for bipartite graph analysis, such as degree-based model $(\alpha,\beta)$-core and butterfly-based model bitruss. 
    All these models have various applications in different domains~\cite{DBLP:conf/www/LiuYLQZZ19,wang2022towards}. In this paper, we focus more on the strong connections among entities (e.g., co-purchase) in the target layer. Thus, we choose biclique, which is the most cohesive model. Recall the example in Figure~\ref{fig:intro}, if we directly replace the biclique constraint with (2,2)-core in our model, all users $\{v_1,v_2,v_3,v_4,v_5\}$ will be returned as a group. 
    
In this paper, we employ the frequency constraint (i.e., $\lambda$), since frequently appearing patterns are usually worthy of attention and may represent the important concept in the environment~\cite{zhang2023discovering,aslay2018mining,yang2016diversified,yang2004complexity,yang2006computational}, such as frequent co-purchasing behavior discussed above. 
Besides, compared with the existing temporal models that mainly focus on interval-based or periodic-based timestamp constraints, 
our model can better capture the real-world patterns occurring frequently.
Moreover, we employ the maximality constraint, since any subset of an \mrg with size no less than $\tau_V$ is also a $\lambda$-\frequent group. Without the maximality constraint, it may generate many redundant results. 
In this paper, we aim to enumerate all \mrgs from a temporal bipartite graph.

In addition to the application of customer analysis mentioned above, \mrg can find many other applications in different domains. 
For instance, a temporal bipartite graph is a suitable data structure to model the data of patients' diagnostic records, where the vertices correspond to patients (i.e., $U$) and health conditions (i.e., $V$), and the links indicate the presence of a diagnosis at the corresponding timestamp~\cite{aziz2021multimorbidity,zhao2019integrating,krishnagopal2020identifying}.  
By mining \mrgs from the patient-condition temporal bipartite graph, we can find the combinations of health conditions that frequently and simultaneously appear in multiple patients. The results can provide data support for the study of multimorbidity, facilitating diagnosis and prevention \cite{schafer2014reducing,vetrano2020twelve}.
In Section~\ref{exp}, we present two case studies on real-world datasets to illustrate the effectiveness of our model.

\myparagraph{Challenges and our approaches}
To the best of our knowledge, we are the first to propose and investigate the maximal $\lambda$-frequency group (\mrg) enumeration problem in temporal bipartite graphs.
We prove the hardness of counting \mrgs.     
    In the literature, maximal biclique enumeration is the most relevant problem to ours (\eg~\cite{abidi2020pivot,chen2022efficient,zhang2014finding,yao2022identifying}).
    However, the introduction of temporal and unilateral aspects significantly complicates the problem.
    Naively, we can enumerate all bicliques over each snapshot and post-process the intermediate results. 
    Due to the hardness of the biclique enumeration problem in static bipartite graphs, treating each timestamp separately is time-consuming.
    Moreover, since the correlation among vertices varies over time, considering temporal and cohesive aspects simultaneously in algorithm design is nontrivial.
In previous studies, the Bron\text{-}Kerbosch (BK) framework is widely used for biclique enumeration (\eg~\cite{abidi2020pivot,chen2022efficient}). It iteratively adds vertices from the candidate set to expand the current result in a DFS manner for biclique enumeration. 
    By extending the BK framework, in our problem, we need to further check the frequency and maximality constraints for each candidate group.
    Since \mrg focuses on the unilateral vertex set and the number of bicliques in each timestamp is large, it means numerous candidate groups will be generated while only very few of them will belong to the final results. 
    Thus, the huge time cost to apply the naive frequency verification on each candidate group and eliminate the non-maximal results presents a unique challenge for our problem.

To address the challenges, in this paper, we first propose a filter-and-verification (\filterv) approach by leveraging the BK framework. 
Generally, \filterv maintains a recursive search tree and traverses in a depth-first manner.
In each iteration, \filterv iterates over all the candidate vertices (i.e., vertex-oriented search paradigm), verifies the frequency constraint and obtains the \textit{valid candidate set}.
Note that, for each vertex in the valid candidate set, the group composed of it and the current processing set still meets the frequency requirement.
If no more vertices can be added to the current processing set to form a new frequent group, \filterv terminates the current search branch and checks the maximality of the vertex set.
Since there could be many vertices that cannot be involved in any \mrgs, a novel structure 
$(\tau_V,\tau_U,\lambda)$-core is first designed to reduce the search space.
Then, a filter strategy is proposed to first efficiently prune the candidate set before the examination, which can reduce the unnecessary call of frequency verification.
To further accelerate the frequency verification of a given vertex set, we present an elaborate \twoarr-based verification strategy.
The frequency check method is also employed to accelerate the maximality verification by avoiding the numerous set comparisons.

Even though \filterv can remarkably accelerate the \mrg enumeration procedure, we need to compute the valid candidate set for each current processing result during the search. 
The overall cost significantly increases with the size of candidate set and the workload of maximality verification, which may hinder its scalability to larger graphs.
As shown in Table~\ref{tab:timeofgcs}, whose details can be found in Section~\ref{sec:method2}, the components of the valid candidate set computation and maximality verification take up a majority of the overall execution time.
Therefore, if we can reduce or even avoid the cost of frequency and maximality verification to some extent, the overall performance can be significantly improved.
Motivated by this,
we further develop a novel verification-free (\vfree) approach. 
Instead of iterating over vertices during the candidate set computation and maximality verification, 
we develop  
a timestamp-oriented search paradigm.
That is, \vfree iterates through the timestamps to obtain the valid candidate set using the advanced dynamic counting structures proposed, 
where the unpromising timestamp can be skipped and common neighbor information can be carried forward in the subsequent search process.
Theoretically, we prove that \vfree can significantly
reduce each valid candidate set computation cost in \filterv by a factor of $\mathcal{O}(|V|)$.
Additionally, by integrating the developed search paradigm and dynamic counting techniques, \vfree can avoid explicit maximality verification.

\myparagraph{Contributions} The main contributions of the paper are summarized as follows.


\begin{itemize}[leftmargin=*,topsep=0pt]

\item  To capture the properties of temporal bipartite graphs, we conduct the first research to propose and investigate the maximal $\lambda$-\frequent group enumeration problem.  \hfill (Section~\ref{sec:pre})

\item To solve the problem, we introduce a filter-and-verification framework. Novel $(\tau_V,\tau_U,\lambda)$-core structure and candidate filtering rule are developed to shrink the search space. Advanced \twoarr-based method is proposed to accelerate the computation of valid candidate set and maximality verification. \hfill (Section~\ref{sec:method1})

\item To overcome the frequency verification cost and scale for larger networks, we further develop a verification-free framework by leveraging the dynamic counting structure proposed. The framework can also avoid the explicit verification of maximality based on the propounded search paradigm. \hfill (Section~\ref{sec:method2})

\item Extensive experiments are conducted on 15 real-world graphs to demonstrate the performance of proposed techniques and model. Compared with the baseline, the optimized method can achieve up to three orders of magnitude speedup. 
\hfill (Section~\ref{exp})

\end{itemize}

\section{{Preliminary and Problem Definition}}
\label{sec:pre}


\begin{figure}
	\centering
	\includegraphics[width=\linewidth]{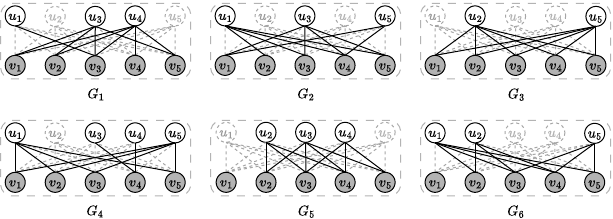}
	\caption{{A temporal bipartite graph $\mathcal{G}$ with six timestamps ($G_1$-$G_6$ are the corresponding snapshots and solid lines denote the edges in each snapshot)}}
	\label{fig:tbg}
\end{figure}

\subsection{Preliminary} 
Let $\mathcal{G}=(U,V,\mathcal{E})$ denote an undirected temporal bipartite graph, 
where $U$ and $V$ are two disjoint vertex sets, \ie $U \cap V = \emptyset$, 
and $\mathcal{E}$ is the set of temporal edges.
$(u,v,t)$ denotes a temporal edge between $u\in U$ and $v\in V$, where 
$t$ is the interaction timestamp between $u$ and $v$. 
Without loss of generality, 
we use $\mathcal{T}=\{t_1,t_2,\dots,t_{|\mathcal{T}|}\}$ to represent the set of timestamps, i.e., 
$\mathcal{T}=\{t|(u,v,t)\in \mathcal{E}\}$\footnote{We use the same setting as the previous studies for the timestamp, which is the integer, since the UNIX timestamps are integers in practice \cite{qin2022mining,zhang2023discovering}.}. 
Given a temporal bipartite graph $\mathcal{G}$, its corresponding \textit{static bipartite graph}, i.e., by ignoring all the timestamps on edges, is denoted by $G=(U,V,E)$, where $E=\{(u,v)|(u,v,t) \in \mathcal{E}\}$.
We can extract a series of snapshots $\{G_1,G_2,\dots,G_{\mathcal{|T|}}\}$ from $\mathcal{G}$ based on the timestamps.
Specifically, given a timestamp $t\in \mathcal{T}$, its corresponding snapshot is a bipartite graph $G_t =(U_t, V_t, E_t)$, where $U_t =\{ u|(u,v,t)\in \mathcal{E}\}$, $V_t =\{ v|(u,v,t)\in \mathcal{E}\}$, and $E_t =\{(u,v)| (u,v,t) \in \mathcal{E}\}$.

\begin{definition} [Structural neighbor (\sstrneighbor)] 
Given a vertex $u\in \mathcal{G}$, the \textit{\sstrneighbor} set of $u$ is the set of vertices connected to $u$ in 
$G$, denoted by $\strneighborset{u}$, i.e., $\strneighborset{u} = \{v| (u, v) \in E\}$. 
$\strneighbordeg{u}$ denotes its \textit{\fstrdegree} (\textit{\sstrdegree}), i.e., $\strneighbordeg{u}$ $= |\strneighborset{u}|$.
\end{definition}

\begin{definition}[Momentary neighbor (\smomeneighbor)]
\label{mn}
Given a vertex $u \in \mathcal{G}$ and a timestamp $t \in \mathcal{T}$, the \textit{\smomeneighbor} set of $u$ at $t$ is the set of vertices connected to $u$ in $G_t$, denoted by $\momeneighborset{u}{t}$, \ie $\momeneighborset{u}{t} = \{v|(u,v) \in E_t\}$. We use $\momeneighbordeg{u}{t}$ to denote its \textit{\fmomedegree} (\textit{\smomedegree}) at $t$, \ie $\momeneighbordeg{u}{t} = |\momeneighborset{u}{t}|$.
\end{definition}

\begin{example}
    {Figure \ref{fig:tbg} shows a temporal bipartite graph with six snapshots $G_1$ to $G_6$.}
    The \sstrneighbors of $u_1$ are $\{v_1, v_2, v_3, v_4, v_5\}$ and the \sstrdegree $\strneighbordeg{u_1}$ is 5. The \smomeneighbor of $u_1$ at $t=1$ is $\{v_3\}$ and the corresponding \smomedegree $\momeneighbordeg{u_1}{1}$ is 1.
\end{example}

Given a static bipartite graph $G = (U,V, E)$, a biclique $S=(U_S, V_S, E_S)$ is a complete subgraph of $G$, where $U_S \subseteq U$, $V_S \subseteq V$, and for each pair of vertices $u\in U_S$ and $v\in V_S$, we have $(u,v) \in E_S$.

\begin{definition}[$(\tau_U, \tau_V)$-biclique]
Given a static bipartite graph $G$ and two positive integers $\tau_U$ and $\tau_V$, a $(\tau_U, \tau_V)$-biclique $S=(U_S,$ $V_S,E_S)$ is a biclique of $G$ with $|U_S|\geq\tau_U$ and $|V_S|\geq\tau_V$.
\end{definition}

\subsection{Problem Definition}

In this paper,
we aim to retrieve the frequent vertex set in unilateral layer, \eg the customer set in Figure~\ref{fig:intro}.
\textbf{For presentation simplicity, we assume the frequent vertex set is from $V$.}
{Before introducing the \rg, we first define the support timestamp for the unilateral vertex set below.}

\begin{definition} [Support timestamp]
\label{def:st}
Given a temporal bipartite graph $\mathcal{G}=(U,V,\mathcal{E})$, a subset $V_S \subseteq V$, a snapshot $G_t =(U_t, V_t, E_t)$ of the timestamp $t \in \mathcal{T}$, and two positive integers $\tau_U$, $\tau_V$, we say $t$ is a support timestamp of $V_S$, 
if $i)$ $V_S \subseteq V_t$, and $ii)$ $V_S$ can form a $(\tau_U, \tau_V)$-biclique with a subset of vertices in $U_t$, 
\ie $V_S$ is included in a $(\tau_U, \tau_V)$-biclique of $G_t$.
\end{definition}

\begin{definition}[\rg]\label{def:rg}
Given a temporal bipartite graph $\mathcal{G}=(U,V,\mathcal{E})$, and three positive integers $\tau_U$, $\tau_V$ and $\lambda$, a \rg is a subset of vertices $V_S \subseteq V$ where there are at least $\lambda$ support timestamps in $\mathcal{T}$ for $V_S$.
\end{definition}

\begin{definition}[Maximal \rg (\textbf{\mrg})]\label{def:mrg}
Given a temporal bipartite graph $\mathcal{G}=(U,V,\mathcal{E})$ and three positive integers $\tau_U$, $\tau_V$ and $\lambda$, a \rg $V_S$ is maximal if there is no other \rg $V'_S$ that is a superset of $V_S$.
\end{definition}

\noindent \textbf{Problem definition.} 
Given a temporal bipartite graph $\mathcal{G}$ and 
three positive integers $\tau_U$, $\tau_V$ and $\lambda$, 
in this paper, we aim to find all the maximal $\lambda$-frequency groups (\mrgs) in $\mathcal{G}$.

\begin{example}
\label{fig:pre}
Considering the temporal bipartite graph shown in Figure \ref{fig:tbg}, suppose $\tau_U=\tau_V=2$ and $\lambda=3$. 
There are three \mrgs in the graph,
\ie $V_{S1}=\{v_1,v_2,v_3,v_5\}$ with 3 support timestamps $\{t_1,t_3,t_4\}$, $V_{S2}=\{v_2, v_3,v_4\}$ with 3 support timestamps $\{t_3,t_5,t_6\}$ and $V_{S3}=\{v_3,v_4,v_5\}$ with 4 support timestamps $\{t_2,t_3,t_5,t_6\}$.
\end{example}

\noindent \textbf{Problem properties.}
Based on the definition of \mrg, Lemmas~\ref{l1} and \ref{l2} can be immediately obtained. The proofs are omitted. 
Then, we show the hardness of our problem in Theorem~\ref{np}.

\begin{lemma}[Structural property]
\label{l1}
Given a temporal bipartite graph $\mathcal{G}=(U,V,\mathcal{E})$ and three positive integers $\tau_U$, $\tau_V$ and $\lambda$, any \mrg
$V_S \subseteq V$ must be contained in a maximal $(\tau_U,\tau_V)$-biclique of the static bipartite graph $G$ of $\mathcal{G}$.
\end{lemma}

\begin{lemma}[Antimonotone property]
\label{l2}
Given a vertex set $V_S \subseteq V$, $i)$ if $V_S$ satisfies the frequency constraint, any subset of $V_S$ also satisfies the constraint; 
$ii)$ if $V_S$ does not meet the frequency constraint, any superset of $V_S$ is not frequent. 
\end{lemma}

\begin{theorem}\label{np}
The problem of counting all \mrgs is \#P-complete.
\end{theorem}

\begin{proof}
We reduce a well-known {\#P-complete} 
problem of {counting}  
maximal $\sigma$-frequency itemsets~\cite{yang2004complexity,yang2006computational} to the \mrgs~{counting}  
problem with $\tau_U \geq 1$, $\tau_V \geq 1$ and $\lambda = \sigma$. 
Let $X = \{x_1, x_2, \dots, x_n\}$ be a set of items. 
Given a database $\mathcal{D}$ consisting of $|\mathcal{T}|$ transactions, \ie $\mathcal{D}=\{d_1, d_2, \dots, d_{|\mathcal{T}|}\}$. 
Each transaction $d_i$ comprises several items in $X$, \ie $d_i \subseteq X$. 
The maximal $\sigma$-frequency itemsets problem is to enumerate all subsets $\{I_1, I_2, \dots, I_k, \dots\}$, where each subset $I_k$ satisfies that $i)$ $I_k \subseteq X$, $ii)$ $|\{d_i \in \mathcal{D}| I_k \subseteq d_i\}| \geq \sigma$ and $iii)$ any superset of $I_k$ does not meet $i)$ and $ii)$.

We construct an instance of a temporal bipartite graph $\mathcal{G} = (U, V, \mathcal{E})$ from $X$ and $\mathcal{D}$ to prove the hardness.
First, we generate $\tau_U$ vertices to form set $U$.
Second, for each item $x_j \in X$, a vertex $v_j$ is correspondingly created to form set $V$, i.e., $|V| = |X|$. 
Let $V_i \subseteq V$ be the set of vertices, where each vertex $v_j$ corresponds to each item $x_j$ in the transaction $d_i \in \mathcal{D}$, $|V_i|=|d_i|$. 
Third, for each transaction $d_i$, we generate a timestamp $t_i$ such that each $v_j \in V_i$ connects to all vertices in $U$ at this timestamp. Then, there will be $|\mathcal{T}|$ timestamps in $\mathcal{G}$, and $(U, V_i)$ is the only one biclique at the timestamp $t_i$.

We then show that this transformation of $X$ and $\mathcal{D}$ into $\mathcal{G}$ is a reduction. Suppose $\{I_1, I_2, \dots, I_k, \dots\}$ {is} the set of all the maximal $\sigma$-frequency itemsets w.r.t $\mathcal{D}$ in $X$. 
We claim that the corresponding vertex sets $\{V_1, V_2, \dots, V_k, \dots\}$ is the set of {all \mrgs in $\mathcal{G}$}. 
Take the vertex set $V_k$ as an example.
We first prove that $V_k$ is a \rg. 
According to the above construction, $V_k$ can form a biclique with $U$ at no less than $\lambda$ timestamps, since $I_k$ must be the subset of at least $\sigma$ transactions in $\mathcal{D}$ and $\lambda = \sigma$. 
Second, we prove the maximality of the $V_k$.
Suppose to the contrary that $V_k$ is not maximal (\ie there is a vertex $v_j$ that can join $V_k$ to form a new \mrg), $I_k$ is not maximal, since $I_k \cup \{x_j\}$ will be a $\sigma$-frequency itemset.
This contradicts the condition that $I_k$ is maximal.
Third, if there exists an \mrg $V_z$ that is not included in $\{V_1, V_2, \dots, V_k, \dots\}$, there must exist the other itemset $I_z \notin \{I_1, I_2, \dots, I_k, \dots\}$ that is $\sigma$-frequency.
Therefore, $\{V_1, V_2, \dots, V_k, \dots\}$ is the complete set of all \mrgs in $\mathcal{G}$.
Conversely, assume that $\{V'_1, V'_2, \dots, V'_k, \dots\}$ are the set of all \mrgs in $\mathcal{G}$. 
Based on the definition of \mrg and the above construction, the one-to-one correspondence between the maximal $\sigma$-frequency itemsets of $\mathcal{D}$ and the \mrg in $\mathcal{G}$ is established.
Therefore, if we count the number of all $\sigma$-frequency itemsets of $\mathcal{D}$, we can obtain the number of all \mrgs. The reduction is realized.
Therefore, the problem of counting the number of \mrgs is \#P-complete.
\end{proof}

\vspace{-1.3mm}
Note that, our enumeration problem is at least as hard as the counting problem, because if we can enumerate all the results, then we can easily count the total number.

\section{filter-and-verification approach}
\label{sec:method1}

\begin{algorithm}[t]
{
    \SetVline
    \footnotesize
    \caption{\textsf{Filter-and-Verification (\filterv)}}
    \label{alg:mfgl}
    \Input{$\mathcal{G}=(U,V,\mathcal{E})$: a temporal bipartite graph, \\ $\tau_U$, $\tau_V$: size constraints, $\lambda$: frequency constraint}
    \Output{$\mathcal{R}$: all the \mrgs}
  
    \StateCmt{$\mathcal{G} \gets$ \textsf{GFCore}$(\mathcal{G},\tau_U,\tau_V,\lambda)$}{graph filter: Algorithm \ref{alg:abocore}}
    \State{$\mathcal{R} \gets \emptyset$}
    \State{\textsf{EnumMFG}$(U, \emptyset, V)$}     
    \vspace{1mm}
  
    {\textbf{Procedure} \textsf{EnumMFG}$(U_S, V_S, C_V)$}\\
    \StateCmt{\textsf{$C_V \gets$ \text{filter candidate set} $C_V$}}{candidate filter: Lemma \ref{lemma-a2-filter}}
    \StateCmt{$C^*_V \gets \emptyset$}{{compute valid candidate set}}
    \ForEach{$v \in C_V$}
    {
        \If(\tcc*[f]{Algorithm \ref{alg:fc}}){\textup{\textsf{CheckFRE}(}$U_S\cap \strneighborset{v}, V_S\cup\{v\},\lambda$\textup{)}}
        {
            \State{$C^*_V \gets C^*_V \cup \{v\}$}
        }
    }
    \If{$|U_S| < \tau_U$ $\lor$ $|V_S|+|C^*_V| < \tau_V$}
    {
        \State{\textbf{return}}
    }
    \If{$C^*_V = \emptyset$} 
    { 
        \StateCmt{\text{Check the maximality for $V_S$}}{{Section~\ref{subsec:mal}}}
                
        \State{\textbf{if} $V_S$ \text{is maximal} \textbf{then} $\mathcal{R} \gets \mathcal{R} \cup \{V_S\}$}             
    }
    \ForEach{$v\in C^*_V$}
    {
        \State{$C^*_V \gets C^*_V \backslash \{v\}$}
        \State{\textsf{EnumMFG}$(U_S \cap N(v,G), V_S \cup \{v\}, C^*_V)$}
    }             

}
\end{algorithm}

In the literature, the closest problem to ours is the maximal biclique enumeration problem (e.g.,~\cite{abidi2020pivot,chen2022efficient,zhang2014finding}),
where most studies are based on the Bron\text{-}Kerbosch (BK) framework.
It maintains a recursion search tree and traverses in a depth-first manner. 

\myparagraph{Baseline method}
Motivated by Lemmas \ref{l1} and \ref{l2}, a reasonable approach for our problem is to employ the BK framework by jointly considering the frequency and maximality constraints. 
Specifically, we operate on three dynamically changing vertex sets $(U_S,V_S,C_V)$.
$V_S$ is the current result.
$U_S$ is the common \sstrneighbors of all the vertices in $V_S$. 
$C_V$ is the candidate set. 
In each iteration, we select a vertex from the candidate set $C_V$ to expand $V_S$, and update the corresponding $U_S$.
If $V_S$ satisfies the frequency constraint, we continue to expand it.
If no other vertex can be added into $V_S$ to form a new frequent group,
we terminate the current search branch and check the 
maximality of $V_S$ by comparing it with the existing found results.
After enumerating through each search branch, all the \mrgs are returned. 
This algorithm is referred to as \bkalg.


\myparagraph{Limitations}
Although \bkalg can correctly return all the \mrgs for a given temporal bipartite graph, we find that directly extending the BK framework is inefficient due to the following two drawbacks.
The first drawback is the huge search space.
The search space of \bkalg is the whole graph $\mathcal{G}$, and it needs to iterate through all the vertices in $C_V$ in each branch, which may involve many unpromising vertices that cannot exist in any \mrg. 
The second drawback is the cumbersome frequency {constraint} check. 
Similar to biclique enumeration in static graphs, in \bkalg, $C_V$ maintains the candidate vertices and we need to ensure the frequency constraint during the search, which is computationally expensive.

To address these limitations, in this section, 
we propose a novel filter-and-verification (\filterv) algorithm.
In the following, we first introduce the search framework of \filterv (Section \ref{subsec:framework}).
For \textit{drawback 1}, we develop novel graph and candidate set filtering techniques to dramatically shrink search space (Section \ref{subsec:fliter}). 
For \textit{drawback 2}, an advanced \twoarr-based algorithm is presented to facilitate the frequency (Section \ref{subsec:verify}) and maximality (Section~\ref{subsec:mal}) verification.  
\vspace{-5mm}

\subsection{Framework Overview}
\label{subsec:framework}

Hereafter we present an overview of our filter-and-verification (\textbf{\filterv}) framework.
We call each search branch that fails to find an \mrg an invalid branch.
Recall the search branch $(U_S,V_S,C_V)$.
To reduce even avoid the search cost on invalid search branches, instead of directly adding each candidate vertex from $C_V$ into $V_S$, we first perform a verification procedure on $C_V$ to generate the \textbf{valid candidate set} $C^*_V \subseteq C_V$ for $V_S$.
That is, in subsequent branches, the new set obtained by adding {any} candidate vertex {from $C_V^*$} to the current processing set $V_S$ still can meet the frequency requirement.


\myparagraph{\filterv framework}
Motivated by the above idea, the pseudocode of \filterv framework is presented in Algorithm~\ref{alg:mfgl}.
It first applies the graph filtering technique (Algorithm \ref{alg:abocore} in Section \ref{subsec:fliter}) to shrink the given temporal bipartite graph.
In the following, \filterv invokes the procedure \textsf{EnumMFG} to enumerate all the \mrgs.
Similar as the BK method, \textsf{EnumMFG} maintains three sets $U_S$, $V_S$ and $C_V$, which are initialized as $U$, $\emptyset$ and $V$.
In \textsf{EnumMFG}, we first try to filter the candidate set $C_V$ (line 5) and then compute the valid candidate set $C^*_V$ for $V_S$ (lines 7-9). 
The branch is terminated if it violates the $(\tau_U,\tau_V)$-biclique size constraints.
We check the maximality of $V_S$ when the {valid} candidate set is empty (lines 12-14).
If $C^*_V$ is not empty, we process each vertex $v\in C^*_V$ to expand $V_S$, and continue search on the updated $V_S$ and $U_S$  (lines 15-17). 


\myparagraph{Discussion}
To compute the valid candidate set $C^*_V$ for $V_S$ in lines 7-9, we need to check the frequency of each vertex set $V_S \cup \{v\}$ for $v\in C_V$.
Given the vertex set $V_S$ and the checking vertex $v\in C_V$, a naive method to check the frequency of $V_S \cup \{v\}$ is to compute the common \smomeneighbors of $V_S\cup\{v\}$ at each timestamp.
Specifically, for each timestamp $t\in \mathcal{T}$, {we check whether there exists no less than $\tau_U$ common \smomeneighbors of all the vertices in $V_S\cup\{v\}$.}
If it satisfies the constraint, $t$ can contribute to the frequency for $V_S\cup\{v\}$.
If the number of such timestamps for $V_S\cup\{v\}$ is no less than $\lambda$, $v$ is the valid candidate vertex for $V_S$ and it can be added into $C^*_V$.
After checking all vertices in $C_V$, we can return $C^*_V$.
Due to the large scale of candidate vertices and the inefficiency of the naive frequency checking method, the above process is very time-consuming.
To speed up the computation of the valid candidate set, we designed novel filtering strategies and efficient verification techniques.

\subsection{Filtering Rules}
\label{subsec:fliter}



\myparagraph{Graph filter} Given a temporal bipartite graph $\mathcal{G}$, many unpromising vertices cannot exist in any \mrgs. 
Therefore, we propose a novel graph structure to filter the search space.
Before presenting the details, we first introduce the concept of $(\alpha,\beta)$-core~\cite{DBLP:conf/www/LiuYLQZZ19}.

\begin{definition}[$(\alpha,\beta)$-core]
Given a static bipartite graph $G$ and two positive integers $\alpha$ and $\beta$, the subgraph $S=(U_S,V_S,E_S)$ is the $(\alpha,\beta)$-core of $G$, denoted by $\abcore{G}$, if it satisfies: $i)$ the degree of each vertex in $U_S$ is at least $\alpha$ and the degree of each vertex in $V_S$ is at least $\beta$ in $S$, and $ii)$ any supergraph of $S$ cannot satisfy $i)$.
\end{definition}

To compute the $(\alpha,\beta)$-core, we can iteratively remove the vertices that violate the degree constraint with time complexity $\mathcal{O}(|E|)$.
Based on the definition and properties of \mrg, we can obtain that every vertex of \mrg must be contained in the temporal bipartite graph $\{\abcore{G_1},\abcore{G_2},\dots,\abcore{G_{|\mathcal{T}|}}\}$.
To further model the frequency property of the vertex, we present a new frequent cohesive subgraph model, called $(\tau_V, \tau_U, \lambda)$-core, which will be applied to prune unpromising vertices before enumerating all the \mrgs.
\begin{definition}[$(\tau_V, \tau_U, \lambda)$-core]
\label{def:newcore}
Given a temporal bipartite graph $\mathcal{G}$ and three positive integers $\tau_U$, $\tau_V$ and $\lambda$, 
the induced subgraph $(U_S,V_S, \mathcal{E}_S)$ of $\mathcal{G}$ is the $(\tau_V, \tau_U, \lambda)$-core if it meets the following conditions: 
$i)$ each vertex $u\in U_S$ can be included in the $(\tau_V, \tau_U)$-core of at least one snapshot,
$ii)$ each vertex $v\in V_S$ can be included in the $(\tau_V, \tau_U)$-core of at least $\lambda$ snapshots, and
$iii)$ there is no supergraph of $(U_S,V_S,\mathcal{E}_S)$ that satisfies $i)$ and $ii)$.
\end{definition}

{Note that, in $(\tau_U, \tau_V)$-biclique, $\tau_U$ and $\tau_V$ restrict the number of vertices in $U$ and $V$. But, in Definition \ref{def:newcore}, the parameters restrict the number of neighbors, so the order is changed.}
Based on the analysis, we can directly obtain the connection between the $(\tau_V, \tau_U, \lambda)$-core and the proposed \mrg model, detailed in the following lemma.

\begin{lemma} 
\label{lemma:core}
Given a temporal bipartite graph $\mathcal{G}=(U,V,\mathcal{E})$ and three positive integers $\tau_U$, $\tau_V$ and $\lambda$, 
if a subset $V_S \subseteq V$ is an \mrg, then $V_S$ must be contained in the $(\tau_V, \tau_U, \lambda)$-core of $\mathcal{G}$.
\end{lemma}

\begin{algorithm}[t]
\SetKwInOut{Input}{Input}
\SetKwInOut{Output}{Output}
{
    \SetVline
    \footnotesize
    \caption{\textsf{GFCore}$(\mathcal{G},\tau_U,\tau_V,\lambda)$}
    \label{alg:abocore}
    \Input{$\mathcal{G}=(U,V,\mathcal{E})$: a temporal bipartite graph, \\ $\tau_U$, $\tau_V$: size constraints, $\lambda$: frequency constraint}
    \Output{$(\tau_V, \tau_U, \lambda)$-core of $\mathcal{G}$}
    
    \State{\textbf{for each} $w\in \mathcal{G}$ 
 \textbf{do} $s[w]=0$}
    \ForEach{$t\in \mathcal{T}$}{
        \ForEach{$w\in \mathcal{G}$}{
            \State{\textbf{if} $w\in U \land \delta(w,t)>0$ \textbf{then} $s[w]++$}
            \State{\textbf{if} $w\in V \land \delta(w,t)>0$ \textbf{then} $s[w]++$}
        }
    }
    \ForEach{$t\in \mathcal{T}$}{
        \ForEach{$w\in \mathcal{G}$}{
            \If{$w\in U \land 0<\delta(w,t)<\tau_V$}{
                \State{\textsf{CorePrune}$(w,t)$}
            }
            \If{$w\in V \land (0<\delta(w,t)<\tau_U \lor 0<s[w]<\lambda)$}{
                \State{\textsf{CorePrune}$(w,t)$}
            }
        }
    }
    \ForEach{$t\in \mathcal{T}$}{
        \ForEach{$w\in \mathcal{G}$}{
            \State{\textbf{if} $\delta(w,t)=0$ \textbf{then} $G_t\gets G_t \backslash \{w\}$}
        }
    }
    \State{\textbf{return} $\{G_1,G_2,\dots,G_{|\mathcal{T}|}\}$}
    
    \vspace{1mm}
    {\textbf{Procedure} \textsf{CorePrune}$(w,t)$}\\
    \State{$\delta(w,t)=0$}
    \ForEach{$x\in \Gamma(w,t)$}{
        \If{$\delta(x,t)>0$}{
            \State{$\delta(x,t)=\delta(x,t)-1$}
            \If{$(x\in U \land \delta(x,t)<\tau_V) \lor (x\in V \land \delta(x,t)<\tau_U)$}{
                \State{\textsf{CorePrune}$(x,t)$}
            }
        }
    }
    \If{$s[w]>0$}{ 
        \State{$s[w]=s[w]-1$}
        \If{$(w \in U \land s[w]<1) \lor (w\in V \land s[w]<\lambda)$}{
            \State{$s[w]=0$}
            \ForEach{$t\in \mathcal{T}$}{
                \If{$\delta(w,t)>0$}{
                    \State{\textsf{CorePrune}$(w,t)$}
                }
            }
        }
    }
}
\end{algorithm}

To derive the $(\tau_V, \tau_U, \lambda)$-core, we can iteratively delete the vertex that violates {the degree constraint or the frequency constraint}.
Since the deletion of one vertex may cause its neighbors to violate the constraints in cascade, we can iteratively prune the graph until all the remaining vertices in $\mathcal{G}$ meet the constraints.
Details of computing $(\tau_V, \tau_U, \lambda)$-core are shown in Algorithm \ref{alg:abocore}.
At first, we use $s[w]$ to count the number of timestamps when the vertex $w$ has enough \smomeneighbors (lines 1-5).
Then we process all the vertices at each timestamp in lines 6-11.
Specifically, for each vertex $w$, if $s[w]$ violates the frequency constraint or the \smomedegree constraint at $t$, we remove this vertex at $t$ and invoke the procedure \textsf{CorePrune} to update the graph.
Details of \textsf{CorePrune} are shown in lines 16-29.
For the processing vertex $w$ at timestamp $t$ in line 16, we set $\delta(w,t)$ as 0 and traverse all \smomeneighbors of $w$ at $t$, \ie $\Gamma(w,t)$.
For each vertex $x\in \Gamma(w,t)$, we first check its \smomedegree.
If it is larger than 0, we reduce its \smomedegree by 1 (lines 19-20).
Then, if $x$ violates the \smomeneighbor constraint at $t$, we invoke \textsf{CorePrune} for $(x,t)$ in lines 21-22.
For the vertex $w$ with $s[w]>0$, we need to update $s[w]$ for it in lines 23-29.
We first reduce $s[w]$ by 1.
If $s[w]$ violates the constraint, we set $s[w]$ as 0 and invoke \textsf{CorePrune} for $w$ at all timestamps (lines 25-29).
After updating the graph, we remove all unsatisfied vertices at each snapshot (lines 12-14).
Finally, we return all the updated snapshots as the reduced graph in line 15.
Similar as $(\alpha,\beta)$-core, the time complexity is bounded by $\mathcal{O}(|\mathcal{E}|)$.  



\myparagraph{Candidate set filter} 
Recall the search branch with processed vertex sets $(U_S,V_S,C_V)$, where we iteratively add one vertex from $C_V$ into $V_S$ and check its frequency. 
If we can efficiently skip a batch of vertices in $C_V$ without compromising any results, we can reduce many unnecessary calls of frequency check. To achieve this, we propose the following rule to quickly filter the candidate vertex set.

\begin{lemma}[{Candidate set filtering rule}]
\label{lemma-a2-filter}
Given a temporal bipartite graph $\mathcal{G}=(U,V,\mathcal{E})$ and $v\in V$, we use  ${T(v)}$ to denote
the set of timestamps when $v$ has more than $\tau_U$ \smomeneighbors, 
\ie ${T(v)}=\{t| t \in$ $\mathcal{T}$ $\land$ $\delta(v,t) \geq \tau_{U}\}$.
Then, for the current processing vertex sets $V_S$, we can skip a candidate vertex $v' \in C_V$, if $|\cap_{v \in V_S\cup\{v'\}} {T}(v)| < \lambda$.
\end{lemma}

\begin{proof} 
If $|\cap_{v \in V_S\cup\{v'\}} {T}(v)| < \lambda$, it means that there exists less than $\lambda$ timestamps, when the number of common \smomeneighbors of $V_S\cup\{v'\}$ is no less than $\tau_U$.
Therefore, $V_S\cup\{v'\}$ is not frequent and we can prune $v'$ from $C_V$, and the lemma holds.
\end{proof}

\subsection{Frequency Verification}
\label{subsec:verify}

Recall the computation of the valid candidate set $C^*_V$ for $V_S$, whose main cost is the frequency verification for all the vertex sets $V_S \cup \{v\}$, where $v\in C_V$. 
In addition, as discussed before, the naive frequency verification method is very time-consuming.
Therefore, reducing the cost of frequency verification is crucial for optimizing the performance of algorithm. Motivated by this, in this section, we design a novel \twoarr-based structure to speed up the processing.
The detailed method is presented in Algorithm \ref{alg:fc}.

\begin{algorithm}[t]
\SetKwInOut{Input}{Input}
\SetKwInOut{Output}{Output}
{
    \SetVline
    \footnotesize
    \caption{\textsf{CheckFRE}$(U_S,V_S,\lambda)$}
    \label{alg:fc}
    \Input{$U_S$: the common \sstrneighbors of all the vertices in $V_S$, \\ $V_S$: the checking vertex set, $\lambda$: frequency constraint}
    \Output{frequency verification result \textit{true}/\textit{false}}

    \State{$\lambda'=0$}
    \StateCmt{\textbf{for each} $t \in \mathcal{T}$ \textbf{do} {UA}$[t]$ = 0}{Update Array}
    
    \ForEach{$u \in U_S$}
     {
            \StateCmt{\textbf{for each} $t \in \mathcal{T}$ \textbf{do} {RA}$[t]$ = 0}{Reborn Array}
            
            \ForEach{$v \in N(u,G) \land v\in V_S$}{
                \ForEach{$t \in \mathcal{T}_{(u,v)}$}{
                    \StateCmt{{RA}$[t]++$}{count $u$'s \smomeneighbors in $V_S$ at $t$}
                }
            }
           \ForEach{$t \in \mathcal{T}$}{
                \If{\textup{{RA}}$[t]$$= |V_S|$}{
                    \StateCmt{{UA}$[t]++$}{count common \smomeneighbors of $V_S$ at $t$}
                }
           }
    }
     \ForEach{$t \in \mathcal{T}$}{
        \If{\textup{{UA}}$[t] \geq \tau_U$}{
            \StateCmt{$\lambda'++$}{count support timestamp for $V_S$}
        }
        \State{\textbf{if} {$\lambda' =\lambda$} \textbf{then return} \textit{true}}
     }
     \State{\textbf{return} \textit{false}}
}
\end{algorithm}

\myparagraph{\textsf{CheckFRE} algorithm}
Algorithm \ref{alg:fc} has three input parameters, i.e., $U_S$, $V_S$ and $\lambda$, which corresponds to line 8 in Algorithm~\ref{alg:mfgl}. It returns \textit{true} if $V_S$ satisfies the frequency constraint. Otherwise, it returns \textit{false}.
The algorithm employs two array structures, i.e., Reborn Array (RA) and {Update Array} ({UA}), both of which have a length of $|\mathcal{T}|$.
{{\underline{\textit{Initialization (lines 1-2).}}}
We initialize $\lambda'$ as 0 to count the {frequency} for the checking vertex set
and all the elements in {UA} as 0.
Then, we process each vertex $u \in U_S$ iteratively (lines 3-10).
{\underline{\textit{Count $u$'s \smomeneighbors in $V_S$ at $t$ (lines 5-7).}}} 
For each processed vertex $u$, we use RA to count its \smomeneighbors in $V_S$ at each timestamp, whose elements are initialized as 0 (line~4).}
$\mathcal{T}_{(u,v)}$ is the set of timestamps associated with {edge} $(u,v)$.
After processing all \smomeneighbors of $u$, the {RA} for $u$ is constructed.
{{\underline{\textit{Count common \smomeneighbors of $V_S$ at $t$ (lines 8-10).}}} 
Then, for each element {RA}$[t]$ in {RA}, we check whether {RA}$[t]$ equals $|V_S|$.
If {RA}$[t] = |V_S|$, it means $u$ connects all vertices in $V_S$ at $t$, 
and we increase the corresponding element {UA}$[t]$ in  {Update Array} by 1.}
{UA}$[t]$ represents the number of the common \smomeneighbors of all vertices in $V_S$ at $t$.
After processing all vertices in $U_S$, we obtain the final {UA}.
{\underline{\textit{Count support timestamp for $V_S$ (lines 11-14).}}} 
Then, we check the {value} of each element in {UA}.
{If there is an element that is no less than $\tau_U$, we add the number of {frequency} $\lambda'$ by 1.
We return \textit{true} if $\lambda'$ equals $\lambda$, which means that $V_S$ is frequent.}
Otherwise, we return \textit{false} (line 15).

\begin{figure}[t]
	\centering
\includegraphics[width=0.85\linewidth]{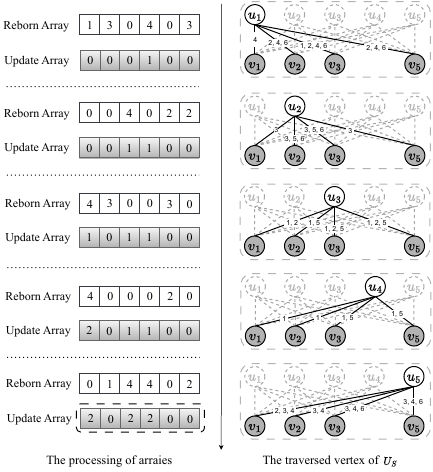}
	\caption{Frequency verification example for vertex set $\{v_1,v_2,v_3,v_5\}$ with $\tau_U=\tau_V=2$ and $\lambda=3$} 
	\label{fig:fc}
	
\end{figure}

\begin{example}
Reconsider the graph in Figure \ref{fig:tbg} with $\tau_U=\tau_V=2$ and $\lambda=3$.
Figure \ref{fig:fc} displays the checking process for the vertex set $V_S=\{v_1,v_2,v_3,v_5\}$, and $U_S=\{u_1,u_2,u_3,u_4,u_5\}$.
Based on Algorithm \ref{alg:fc}, we traverse each vertex in $U_S$ and record its \smomedegree in $V_S$ at each timestamp.
Specifically, when traversing $u_1$, we record its \smomedegree in $V_S$ at each timestamp as $``1,3,0,4,0,3"$ in the {Reborn~Array}.
Among these \smomedegree{s}, only one is equal to $|V_S|=4$ (the fourth element of the {Reborn~Array}), so we add 1 in the fourth element of the {Update~Array}, i.e., {UA}$[4]=1$, indicating that there is one vertex connecting all vertices in $V_S$ at $t=4$.
After traversing $u_1$, we clear {Reborn~Array} and maintain {Update~Array}.
Similarly, we traverse all the vertices in $U_S$ and check all the elements of final {Update~Array}.
We find that there are three elements of the final {Update~Array} that are no less than $\tau_U=2$, i.e., {UA}$[1]=2$, {UA}$[3]=2$, {UA}$[4]=2$, which correspond to the timestamp $t=1$, $t=3$ and $t=4$, respectively.
This means that there are three timestamps at which the vertices in $V_S$ have no less than two common \smomeneighbors, \ie the frequency of $\{v_1,v_2,v_3,v_5\}$ is 3.
Therefore, this vertex set satisfies the frequency constraint.
\end{example}

\myparagraph{Discussion} To compute the valid candidate set, we can iteratively invoke Algorithm~\ref{alg:fc} to examine the frequency of the newly obtained vertex set, i.e., line 8 in Algorithm~\ref{alg:mfgl}. After checking all the vertices in $C_V$, we obtain the valid candidate set $C^*_V$. The corresponding time complexity is shown in Theorem~\ref{tc:fc}.

\begin{theorem}\label{tc:fc}
Based on Algorithm~\ref{alg:fc}, the time complexity of computing the valid candidate set $C^*_V$  is $\mathcal{O}(|V| \cdot d_{max}(u) \cdot d_{max}(v) \cdot |\mathcal{T}|)$, where $d_{max}(u)$ and $d_{max}(v)$ are the largest s-degree of the vertices in $U$ and $V$, respectively.
\end{theorem}

\begin{proof}
In Algorithm~\ref{alg:mfgl}, the main procedure for computing the valid candidate set $C_V^*$ is in lines 7-9. We first analyze the time complexity of Algorithm \ref{alg:fc}. In line 3, the size of $U_S$ is bounded by $d_{max}(v)$ since $(U_S,V_S)$ is a biclique. In lines 5-7, each m-neighbor of $u$ will be visited at most once, which takes $\mathcal{O}(d_{max}(u) \cdot |\mathcal{T}|)$ time. The time complexity of updating {UA} in line 10 is $\mathcal{O}(|\mathcal{T}|)$. Similarly, updating $\lambda'$ also takes $\mathcal{O}(|\mathcal{T}|)$ time. Thus, the time complexity of Algorithm \ref{alg:fc} is $\mathcal{O}(d_{max}(u) \cdot d_{max}(v) \cdot |\mathcal{T}|)$. Moving back to Algorithm~\ref{alg:mfgl}, every vertex in $C_V$ will be checked by Algorithm \ref{alg:fc}, and the size of $C_V$ is bounded by the number of vertex in $V$, i.e., $|V|$. 
Overall, the time complexity for computing the valid candidate set $C_V^*$ is $\mathcal{O}(|V| \cdot d_{max}(u) \cdot d_{max}(v) \cdot |\mathcal{T}|)$.
\end{proof}

\subsection{Maximality Verification}
\label{subsec:mal}

Due to the properties of \mrg, the traditional maximality checking technique for the maximal biclique enumeration~\cite{abidi2020pivot,zhang2014finding,chen2022efficient} cannot be used in our problem.
To check the maximality of obtained vertex set $V_S$, a naive method is to compare it with all the obtained results.
If $V_S$ is the subset of an existing result, \ie not maximal, we can skip it.
Otherwise, for these vertex sets that are the subset of $V_S$, we remove them from the currently found result set and add $V_S$ into the result set.
However, this method requires extensive computation, since it involves numerous set comparisons.
In this section, we introduce a new maximality checking technique.
Specifically, we use the vertex set $X_V$ to store the vertices that are previously processed and can be included in at least one \rg in the current branch, i.e., adding $v$ into $X_V$ after line 17 in Algorithm \ref{alg:mfgl}.
Based on $X_V$, we present the details of checking method below.

\begin{lemma}[{Maximality verification}]\label{lemma:a2-malc}
Given the current processing tuple $(U_S,V_S,C_V,X_V)$, $V_S$ is a \rg.
$V_S$ is an \mrg iff $i)$ $C^*_V=\emptyset$ and
$ii)$ any $v \in X_V$ cannot form a \rg with $V_S$, \ie $\nexists~ v\in X_V~s.t.~V_S \cup \{v\}$ is a \rg.
\end{lemma}


\begin{proof}
If $V_S$ is an \mrg, no vertex from $V\backslash V_S$ can be added into $V_S$ to form a \rg.
The vertices in $V\backslash V_S$ can be categorized as two kinds: one is the non-processed vertex, and the other is the processed vertex.
The non-processed vertices that can form a \rg with $V_S$ will be stored in $C^*_V$ by lines 7-9 of the Algorithm \ref{alg:mfgl}.
Therefore, $C^*_V\neq \emptyset$ means that there exists the other vertex that can join $V_S$ to form a larger \rg and $V_S$ is not maximal.
The vertex processed before and included in at least one \rg are stored in $X_V$.
Therefore, if there is a vertex in $X_V$ that can form a \rg with the current processing vertex set $V_S$, $V_S$ is the subgraph of one found result and $V_S$ is not maximal. The lemma holds.
\end{proof}

By applying Lemma \ref{lemma:a2-malc}, we can eliminate the enormous comparisons to determine the containment relationship between two vertex sets.
The remaining issue is how to efficiently check whether there exists a vertex $v\in X_V$ such that $V_S\cup\{v\}$ is a \rg.
To achieve this, 
we first apply the filtering rule (Lemma \ref{lemma-a2-filter}) {to efficiently shrink $X_V$}. 
Then we apply the verification method  (Algorithm \ref{alg:fc}) to check the frequency of the vertex set obtained by adding each remaining vertex of $X_V$ separately into $V_S$.
\section{Verification-free Approach}
\label{sec:method2}

\filterv is significantly faster than \bkalg, even if we equip \bkalg 
with graph filtering technique (i.e., \bkalgp in the experiment). However, \filterv still suffers from some limitations, due to its search philosophy. 
When computing the valid candidate set $C^*_V$ and checking maximality, \filterv needs to iterate over each vertex in $C_V$ and $X_V$ separately for frequency verification, which could be time-consuming.
In addition, due to the vertex-oriented search paradigm, \filterv cannot effectively utilize the shared information among different computations.
For instance, the invalid timestamp information cannot be inherited by the subsequent computations.
In Table~\ref{tab:timeofgcs}, we report the execution time of computing valid candidate set and checking maximality within \filterv (i.e., \filtervcm(s)), and the percentage \filtervcm(\%) of overall execution time
on a network {D14} with more than 60 million edges (the dataset details can be found in Section~\ref{exp}).
As we can observe, the two components, i.e., compute the valid candidate set and maximality verification, take up a majority of the overall performance.
Therefore, if 
we can reduce the frequency and maximality verification cost or even avoid such cost to some extent, the overall performance can be significantly improved.

This motivates us to develop a strategy without verification, i.e., derive the valid candidate set directly.
Specifically, in this section, we develop a timestamp-oriented search paradigm. 
It iterates through the timestamps to obtain the valid candidate set $C^*_V$ using the dynamic counting structures proposed (Section~\ref{sec:vcc}), 
where the unpromising timestamp can be skipped and common neighbor information can be carried forward.
Then, we present the verification-free algorithm (\vfree) in Section~\ref{sec:va}.
Theoretically, \vfree can significantly reduce each valid candidate set computation cost in \filterv by a factor of $\mathcal{O}(|V|)$ (details are shown in Theorem~\ref{theo:timecomplex-final}). 
Besides, based on the search paradigm proposed in \vfree, we can avoid explicit maximality verification.
In Table~\ref{tab:timeofgcs}, \textsf{\vfreecm(s)} denotes the execution time of computing the valid candidate set and verifying maximality in \vfree, which is much faster than that in \filterv.

\subsection{Valid Candidate Set Computation}
\label{sec:vcc}

\begin{table}[t]
  \caption{Comparison of \filterv and \vfree in computing valid candidate set and checking maximality on D14}
  \label{tab:timeofgcs}
   \small
\begin{tabular}{cccccc}
  \hline
  \textbf{$(\tau_U,\tau_V,\lambda)$} & (8,4,8) & (9,5,8) & (10,6,6) & (10,6,10) \\
  \hline
  \textsf{\filtervcm(\%)} & 88.26\% & 88.52\% & 85.05\% & 86.68\% \\
  \textsf{\filtervcm(s)} & 899.30s & 702.27s & 617.14s & 248.64s \\
  \textsf{\vfreecm(s)} & 63.80s & 28.78s & 26.65s & 9.04s \\
  \hline
\end{tabular}
\end{table}

\myparagraph{General idea} 
In the verification-free framework, we process timestamps sequentially to compute the valid candidate set $C^*_V$. The procedure on one timestamp $t$ consists of four steps, where $V_S$ is the current processing vertex set.

\begin{itemize}[leftmargin=*]
    \item \textit{Step 1: ascertain from $U$}. Obtain the common  \smomeneighbors of $V_S$ in snapshot $G_t$ and store them into $cand_U$.
    
    \item \textit{Step 2: termination check}. If $|cand_U|< \tau_U$, stop processing $t$ and move to the next timestamp, since $V_S$ cannot form a ($\tau_U,\tau_V$)-biclique in $G_t$. Otherwise, $t$ is a \textbf{survived timestamp}. 
    
    \item \textit{Step 3: reverse-ascertain from $V$}. Find all the vertices in $V_t \backslash V_S$ that connect at least $\tau_U$ vertices in $cand_U$ and store them into $cand_V$.
    
    \item \textit{Step 4: survived timestamp update}. 
    Increase the survived timestamp count of $V_S\cup\{v\}$ by 1 for $v\in cand_V$.
\end{itemize}

\begin{example}
\label{exp:method2_main_idea}
Reconsider the graph in Figure \ref{fig:tbg} with $\tau_U=\tau_V=2$.
Suppose $V_S=\{v_1,v_2\}$ and $t=1$. 
\textit{Step 1)} We store the common \smomeneighbors of $V_S$ in $G_1$ to $cand_U$, \ie $cand_U=\{u_3,u_4\}$.
\textit{Step 2)} Since $|cand_U| \geq \tau_U=2$, $t=1$ is a survived timestamp for $V_S$.
\textit{Step 3)} We proceed to examine the common \smomeneighbors of $cand_U$. 
Besides $V_S$, both $v_3$ and $v_5$ connect $\tau_U=2$ vertices in $cand_U$.
Thus, $cand_V = \{v_3,v_5\}$, and we increase the survived timestamp count of $V_S\cup \{v_3\}$ and $V_S\cup \{v_5\}$ by 1, respectively in \textit{Step 4}.
\end{example}

Based on the above procedure, we iterate through all the timestamps.
It is easy to verify that, after processing all the timestamps, $C^*_V$ is the subset of $cand_V$, where for each $v \in C^*_V$, its survived timestamp count is no less than $\lambda$. 
In the following, we present the details about how to $i)$ enable the inheritance of invalid timestamp information and $ii)$ accelerate the neighbor computation.

\myparagraph{Timestamp inheritance} To enable the inheritance of timestamps, we further maintain a timestamp set $C_T$ for $V_S$, which stores all the survived timestamps when computing $C^*_V$ of $V_S$. Then, according to Lemma~\ref{lemma:a3-tsr}, we only need to check the timestamps in $C_T$ when processing the subsequent branches of $V_S$.

\begin{lemma}[Timestamp skipping rule]
\label{lemma:a3-tsr}
Given the processing vertex set $V_S$ and computed $C_T$, 
if $t \notin C_T$, we can skip the processing at the timestamp $t$ in the subsequent search branches of $V_S$, i.e., the branch by adding any vertex into $V_S$.
\end{lemma}

\begin{proof}
By extending the antimonotone property of \mrg (i.e., Lemma \ref{l2}), if $t$ is not the survived timestamp of $V_S$, $t$ cannot be the survived timestamp of any superset of $V_S$. The lemma holds.
\end{proof}

\myparagraph{Neighbor computation acceleration} To compute and maintain the neighbor information, we design three dynamic counting data structures as follows.

\begin{itemize}[leftmargin=*]
    \item A two-dimensional array, denoted by $cnt_U[t][u]$, to count the number of \smomeneighbors of $u\in U$ in $V_S$ at $t\in C_T$.
    \item  A two-dimensional array, denoted by $cnt_V[t][v]$, to count the number of \smomeneighbors of $v\in V\backslash V_S$ in $cand_U$ at $t\in C_T$.
    \item An one-dimensional array, denoted by $cnt_{T}[v]$, used to record the number of survived timestamps for $V_S \cup \{v\}$, where $v\in C_V$.
\end{itemize}

For a given timestamp $t \in C_T$ and the current processing vertex set $V_S$, if $cnt_U[t][u] = |V_S|$, it means $u$ is the common \smomeneighbors of all the vertices in $V_S$ in $G_t$. Therefore, $u$ can be stored in $cand_U$ (correspond to Step 1). If $|cand_U| < \tau_U$, we can skip the timestamp (Step 2). Similarly, if $cnt_V[t][v] \geq \tau_U$, it means $v$ connects to at least $\tau_U$ vertices in $cand_U$ {at $t$} and can be stored in $cand_V$ (Step 3). 
{Then, we increase $cnt_T[v]$ by 1, which means $t$ is a survived timestamp for $V_S \cup \{v\}$ (Step 4)}.

\begin{figure}[t]
\centering
    \begin{minipage}{0.45\textwidth}
        \centering
        \includegraphics[width=\textwidth]{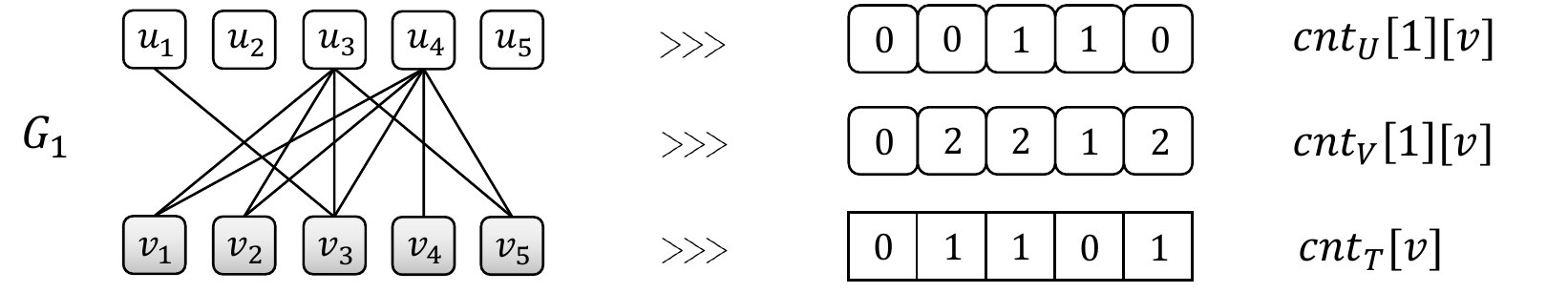}
        \subcaption{$V_S=\{v_1\}$, $C^*_V=\{v_2,v_3,v_5\}, t=1$}
    \end{minipage} \\[0.2cm]
    \begin{minipage}{0.45\textwidth}
        \centering
        \includegraphics[width=\textwidth]{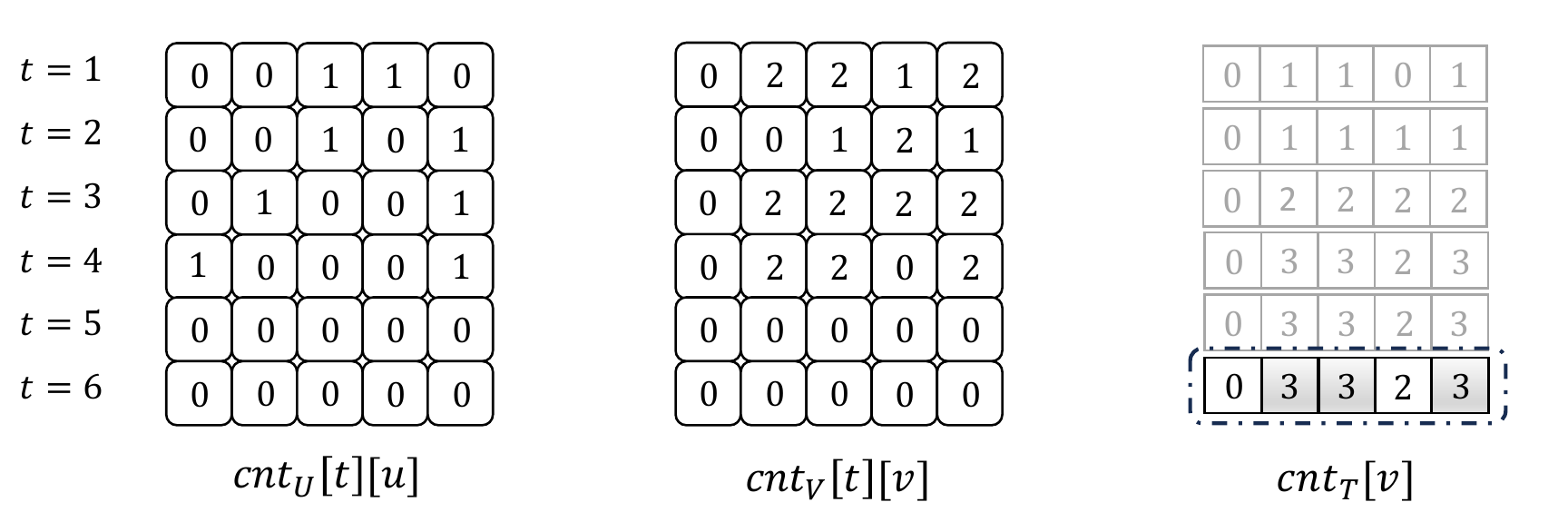}
        \subcaption{$V_S=\{v_1\}$, $C^*_V=\{v_2,v_3,v_5\}$, $C_T=\{1,2,3,4\}$}
    \end{minipage} \\[0.2cm]
    \begin{minipage}{0.45\textwidth}
        \centering
        \begin{subfigure}{\textwidth}
            \includegraphics[width=\textwidth]{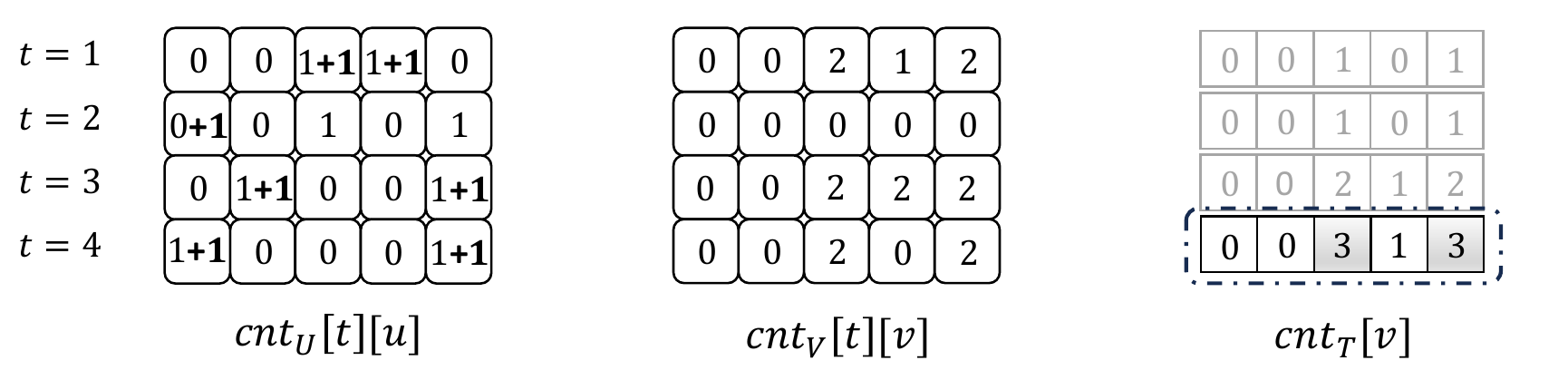}
            \subcaption{$V_S=\{v_1,v_2\}$, $C^*_V=\{v_3,v_5\}$, $C_T=\{1,3,4\}$}
        \end{subfigure}
    \end{minipage}
	\caption{Running example of computing valid candidate set using \vfree  with $\tau_U=\tau_V=2$ and $\lambda=3$}
	\label{fig:new-a3-final}
\end{figure}

\begin{example}
\label{ex:first}
Reconsider the graph in Figure~\ref{fig:tbg} with $\tau_U=\tau_V=2$ and $\lambda=3$. 
Suppose the current processing set $V_S = \{v_1\}$.
To compute the valid candidate set, Figure~\ref{fig:new-a3-final}(a) visualizes the results of three data structures after processing timestamp $t=1$.
When $V_S = \{v_1\}$, we have $C_V=\{v_2,v_3,v_4,v_5\}$ and $C_T=\{1,2,3,4,5,6\}$, since it is the first vertex explored. The three data structures are initialized with 0. The \smomeneighbors of $v_1$ at $t=1$ are $u_3$ and $u_4$. Thus, we increase $cnt_U[1][3]$ and $cnt_U[1][4]$ by 1, respectively. As $cnt_U[1][3] = cnt_U[1][4] = |V_S| =1$, we have $cand_U = \{u_3,u_4\}$. Since $|cand_U| \geq \tau_U = 2$, $t=1$ is a survived timestamp. Next,
we process the \smomeneighbors of $u_3$ at $t=1$, i.e., $\{v_1,v_2,v_3,v_5\}$.
Thus, we assign a value of 1 to $cnt_V[1][2]$, $cnt_V[1][3]$, $cnt_V[1][5]$ for vertex $v_2$, $v_3$ and $v_5$, respectively.
Similarly, we process the \smomeneighbors of $u_4$ and increase the corresponding count. We
have $cnt_V[1][2]=2$, $cnt_V[1][3]=2$, $cnt_V[1][4]=1$ and $cnt_V[1][5]=2$.
Finally, since $cnt_V[1][2]=cnt_V[1][3]=cnt_V[1][5]=2$, we have $cand_V = \{v_2,$ $v_3,v_5\}$ and  increment $cnt_T[2]$, $cnt_T[3]$ and $cnt_T[5]$ by 1.
\end{example}

After processing each timestamp in $C_T$, we get the final $cnt_T[v]$. For a vertex $v' \in V \backslash V_S$, if $cnt_T[v'] \geq \lambda$, it means $v'$ is in the valid candidate set of $V_S$. 

\begin{example}
\label{ex:second}
Following Example~\ref{ex:first}, 
Figure~\ref{fig:new-a3-final}(b) shows the final results of the three data structures after processing all the timestamps in $C_T$. As shown, $cnt_T[v_2]=cnt_T[v_3]=cnt_T[v_5] = 3\geq \lambda$. Therefore, we have the valid candidate set $C_V^* = \{v_2,v_3,v_5\}$, and its survived timestamp set is $\{1,2,3,4\}$. The first column of $cnt_V[t][v]$ is 0. This is because $v_1$ is in the current processing vertex set $V_S$, and we do not need to record the value.
\end{example}

In the search branch of $V_S$, $V_S$ is continuously expanded by adding new vertices into $V_S$. It means that the values in $cnt_U[t][u]$ are non-decreasing. Thus,
we can incrementally maintain the data structure $cnt_U[t][u]$, i.e., $cnt_U[t][u]$ can be inherited in the subsequent search to avoid duplicated computation. 
$cnt_V[t][v]$ and $cnt_T[v]$ are temporally maintained when computing the valid candidate set $C^*_V$ for a $V_S$, and will be reset after each $C^*_V$ computation.

\begin{example}
\label{example:vfmrg2}
Following Example~\ref{ex:second},
suppose we add $v_2$ to expand the vertex set $\{v_1\}$, and compute the valid candidate set $C_V^*$ for $V_S = \{v_1,v_2\}$. Figure~\ref{fig:new-a3-final}(c) shows the final results of the three data structures. The bold values in $cnt_U[t][u]$ denote the incremental computations. Note that, since the survived timestamps of $\{v_1\}$ is $C_T = \{1,2,3,4\}$, we can skip the processing of $t=5$ and 6. Similarly, to compute $C_V^*$, we need to iterate through $C_T$. The difference is that, when updating $cnt_U[t][u]$, we only need to increment the values based on the newly added vertex $v_2$ instead of $V_S$. For instance, at $t=1$, the \smomeneighbors of $v_2$ is $\{u_3, u_4\}$. Therefore, we increment $cnt_U[1][3]$ and $cnt_U[1][4]$ by 1, respectively. 
After processing all the timestamps in $C_T$, 
We have $C^*_V = \{v_3,v_5\}$ by checking $cnt_T[v]$.
\end{example}

\subsection{\vfree Algorithm} 
\label{sec:va}

\begin{algorithm}[t]
{
    \SetVline
    \footnotesize
    \caption{{\textsf{Verification-Free} (\vfree)}}
    \label{alg:final}
    \Input{$\mathcal{G}=(U,V,\mathcal{E})$: a temporal bipartite graph, \\ $\tau_U$, $\tau_V$: size constraints, $\lambda$: frequency constraint}
    \Output{$\mathcal{R}$: all the \mrgs}
    \State{$\mathcal{R}\gets \emptyset$}
    \StateCmt{$\mathcal{G}\gets$ \textsf{GFCore}$(\mathcal{G},\tau_U,\tau_V,\lambda)$}{graph filter: Algorithm \ref{alg:abocore}}
    \State{Reassign vertex id of $V$ in ascending order of the structural degree}
    \ForEach{$t\in \mathcal{T}$}{
        \State{\textbf{for each}~$u\in U$ \textbf{do} $cnt_U[t][u]=0$}
        \State{\textbf{for each}~$v\in V$ \textbf{do} $cnt_V[t][v]=0$}
    }
    \State{\textbf{for each}~$v\in V$ \textbf{do} $cnt_{{T}}[v]=0$}
    \State{\textsf{VerifyFreeMFG}$(\emptyset,V,\mathcal{T})$}

       \vspace{1mm}
       {\textbf{Procedure} \textsf{VerifyFreeMFG}$(V_S,C_V,C_T)$}\\
       \ForEach{$v \in C_V$}{
            \State{$V'_S \gets V_S \cup \{v\},C^*_V\gets \emptyset$, $C'_T\gets \emptyset$, $cand_V\gets \emptyset$}
            \ForEach{$t\in C_T$}{
                \State{$cand_U\gets \emptyset$}
                \StateCmt{$visit_V\gets \emptyset$}{store vertices for the first verify} 
                \ForEach{$u \in \Gamma(v,t)$}{
                    \State{$cnt_U[t][u]=cnt_U[t][u]+1$}
                    \State{\textbf{if} $cnt_U[t][u]=|V'_S|$ \textbf{then} $cand_U.push(u)$}
                }
                \State{\textbf{if} $|cand_U| < \tau_U$ \textbf{then~continue}}
                \StateCmt{$C'_T.push(t)$}{survived timestamp set for $V'_S$}
                \ForEach{$u'\in cand_U$}{
                    \ForEach{$v' \in \Gamma(u',t)$}{
                        \State{\textbf{if} $v' \in V'_S$ \textbf{then~continue}}
                        \If{$v' \notin visit_V$}{
                            \State{$cnt_V[t][v']=1$, $visit_V\gets visit_V\cup \{v'\}$}
                        }
                        \Else{
                        \State{$cnt_V[t][v']=cnt_V[t][v']+1$}
                        }
                        \If{$cnt_V[t][v']=\tau_U$}{
                            \If{$cnt_{{T}}[v']=0$}{
                                \State{$cand_V.push(v')$}
                            }
                            \State{$cnt_{{T}}[v']=cnt_{{T}}[v']+1$}
                        }
                    }
                }
            }
            \State{$\nflag \gets true$}
            \ForEach{$v'\in cand_V$}{
                \State{\textbf{if} $cnt_{{T}}[v']<\lambda$ \textbf{then}~$cnt_T[v']=0$,~\textbf{continue}}
                \State{$cnt_T[v']=0$}
             \State{\textbf{if} $v' < v$ \textbf{then}~$\nflag \gets$ \textit{false}}
             \StateCmt{\textbf{else} $C^*_V.push(v')$}{valid candidate set for $V'_S$}
            }
            \If{$|V'_S|+|C^*_V| \geq \tau_V \land |C'_T|\geq \lambda$}{
                \StateCmt{sort $C^*_V$ based on vertex id}{ensure the processing order}
                \State{\textsf{VerifyFreeMFG}$(V'_S,C^*_V,C'_T)$}
                \StateCmt{\textbf{if} $|C^*_V|=0 \land notRepeat$ \textbf{then}~$\mathcal{R}\gets \mathcal{R}\cup \{V'_S\}$}{\mrg}
            }

            \ForEach{$t\in C_T$}{
                \ForEach{$u\in \Gamma(v,t)$}{
                    \StateCmt{$cnt_U[t][u]=cnt_U[t][u]-1$}{update $cnt_U$}
                }
            }
       }

}
\end{algorithm}

By incorporating the above techniques, we present our verification-free (\vfree) approach, whose details are shown in Algorithm~\ref{alg:final}. 
{\underline{\textit{Initialization (lines 1-7).}} 
$\mathcal{R}$ is used to store all the \mrgs. We also use the graph filtering technique (Algorithm \ref{alg:abocore} in Section \ref{subsec:fliter}) to reduce the search space. 
Then, we reassign the id of each vertex in $V$ in ascending order of the structural degree.}
{Ties are randomly broken if vertices have the same structural degree.}
\textit{Note that}, in \vfree, we process vertices of $V$ in the order of vertex id to ensure the maximality of the result, i.e., avoid explicit maximality verification.
That is,
for any given vertex id setting, as long as we process the vertices in ascending order of the id, the algorithm's properties and correctness can be ensured.
We initialize $cnt_U$, $cnt_V$ and $cnt_T$ in lines 4-7.
In line 8, we invoke the procedure \textsf{VerifyFreeMFG} to enumerate all \mrgs, where we initialize $V_S$ with $\emptyset$, $C_V$ with $V$ and $C_T$ with $\mathcal{T}$.
Generally, $V_S$ is the current processing vertex set, $C_V$ is the {valid} candidate set and $C_T$ is the survived timestamp set for $V_S$. 
Details of the procedure \textsf{VerifyFreeMFG} are shown in lines 9-43. 

In lines 11-36, we compute the valid candidate set for $V'_S = V_S\cup\{v\}$, where the $\nflag$ flag in lines 31 and 35 is used for 
later maximality check. 
We initialize $C^*_V$ , $C'_T$  and $cand_V$ with $\emptyset$ (line 11), and
perform the following operations in turn under each timestamp $t \in C_T$ (lines 12-30).
{\underline{\textit{Step 1: ascertain from $U$ (lines 13-17).}}} 
{We first initialize $cand_U$ with $\emptyset$ to denote the common \smomeneighbor set of all the vertices in $V'_S$ at $t$ 
and $visit_V$ with $\emptyset$ to help maintain the information in $cnt_V[t][v]$.}
Recall that $cnt_V[t][v]$ and $cnt_T[v]$ are temporally maintained, and $cnt_U[t][u]$ can be inherited.
$visit_V$ can help restore the corresponding information in $cnt_V[t][v]$ instead of repeated initialization. 
For each \smomeneighbor $u \in \Gamma(v,t)$, we increase $cnt_U[t][u]$ by 1 denoting that $u$ is connected with $v$ at $t$ (line 16).
If $cnt_U[t][u]=|V'_S|$, which means $u$ connects all the vertices in $V'_S$, we push $u$ into $cand_U$ (line 17).
{\underline{\textit{Step 2: termination check (lines 18-19).}}} 
After processing all the \smomeneighbors of $v$, if $|cand_U|$ is less than $\tau_U$, {which means that $t$ is not a survived timestamp for $V'_S$,  we can skip $t$.
Otherwise, we push $t$ into $C'_T$ that records the survived timestamps for $V'_S$.}
Then we need to explore all the vertices in $cand_U$  in lines 20-30.
{\underline{\textit{Step 3: reverse-ascertain from $V$ (lines 20-29).}}} 
For each vertex $u' \in cand_U$, we need to traverse its \smomeneighbors $\Gamma(u',t)$ iteratively.
Specifically, if its \smomeneighbor $v' \in \Gamma(u',t)$ exists in $V'_S$, we skip the current vertex (line 22).
If $v' \notin visit_V$, which means that it is the first time to visit $v$ in the current search, we set $cnt_V[t][v']$ with 1 and push $v'$ into $visit_V$ (lines 23-24).
Otherwise, we add $cnt_V[t][v']$ by 1 (lines 25-26).
If $cnt_V[t][v']=\tau_U$, which means $v'$ has no less than $\tau_U$ common \smomeneighbors with all vertices in $V'_S$ at timestamp $t$ (\ie $v'$ connects at least $\tau_U$ vertices in $cand_U$), we check whether $cnt_T[v']$ equals 0 or not (lines 27-28).
As discussed before, $cnt_T[v']$ denotes the number of survived timestamps for the set $V'_S \cup \{v'\}$.
If $cnt_T[v']=0$, we push $v'$ into $cand_V$ (lines 28-29).
{\underline{\textit{Step 4: survived timestamp update (line 30).}}} 
{We add the number of survived timestamps for $v'$ (\ie $cnt_T[v']$) by 1.}
{\underline{\textit{Valid candidate set computation (lines 31-36).}}} 
After processing all the timestamps in $C_T$, we can obtain the final $cnt_T[v']$ for each vertex $v'\in cand_V$.
{If $cnt_T[v'] < \lambda$, we set $cnt_T[v']$ as 0 and skip $v'$.
If $v' < v$, we set $notRepeat$ as \textit{false}.
If $v' \geq v$, we push $v'$ into $C^*_V$ as the valid candidate vertex for $V'_S$.}

We recursively invoke \textsf{VerifyFreeMFG} if the size constraint is satisfied (lines 37-39). $C^*_V$ is sorted to ensure the processing order of vertices in $V$.  
We restore $cnt_U[t][u]$ by reducing 1 for each \smomeneighbor $u$ of $v$ to utilize it in the next iteration (lines 41-43).
Unlike \filterv, which explicitly conducts maximality verification through expensive computation (i.e., Section 3.4), in \vfree,  
if $C^*_V=0$ and $notRepeat$ is \textit{true}, it means we find an \mrg that can be added to the result set $\mathcal{R}$ (line 40). 
The correctness is shown in Theorem \ref{theo:correct-final}.





\begin{theorem}[Algorithm correctness]\label{theo:correct-final}
Given a temporal bipartite graph $\mathcal{G}$, three positive integers $\tau_U$, $\tau_V$ and $\lambda$, Algorithm \ref{alg:final} can return all the \mrgs in $\mathcal{G}$.
\end{theorem}

\begin{proof}
We first prove that the search branch in line 39 of Algorithm \ref{alg:final} can return all the \mrgs containing $V'_S$ {that have not been found in the previous search branch}.
Following the algorithm, any vertex $v'\in cand_V$ satisfies $cnt_T[v']\geq \lambda$ is the one that can form a \rg with $V'_S$, \ie $V'_S \cup\{v'\}$ is the \rg.
All these vertices can be enumerated during the search and categorized into two kinds, the one whose id is larger than $v$ and the other whose id is smaller than $v$.
If $v' < v$, which means that $v'$ was processed earlier and $V_S\cup \{v,v'\}$ will be enumerated during the search branch for $V_S\cup\{v'\}$.
If $v' > v$, $v'$ will be added into $C^*_V$.
It is easy to verify that $V'_S$ is not maximal if $|C^*_V|\neq 0$.
Hence, all the \mrgs containing $V'_S$ can be enumerated during the branch in line 39.
Besides, all vertices in $V$ will be traversed in our algorithm, so that \mrg containing each vertex in $V$ can be obtained through the recursion of line 39.
Therefore, the theorem is correct.
\end{proof}

\begin{theorem}[Time complexity analysis]\label{theo:timecomplex-final}
The time complexity of computing the valid candidate set $C_V^*$ for a vertex set $V_S$ based on Algorithm \ref{alg:final} is $\mathcal{O}(d_{max}(u) \cdot d_{max}(v) \cdot |\mathcal{T}|)$.
\end{theorem}


\begin{proof}
In Algorithm \ref{alg:final}, the main procedure of computing the valid candidate set $C^*_V$ are in lines 11-36. The size of survived timestamp set $C_T$ is bounded by $|\mathcal{T}|$ in line 12, \ie $|C_T| \leq |\mathcal{T}|$. 
The time complexity of generating $cand_U$ at each timestamp $t$ is $\mathcal{O}(d_{max}(v))$ in lines 15-17, and $|cand_U| \leq d_{max}(v)$. Lines 22-30 can be done in constant time, and are performed $d_{max}(u) \cdot d_{max}(v)$ times at each timestamp $t$. Thus, the time complexity of running lines 12-30 is $\mathcal{O}(d_{max}(u) \cdot d_{max}(v) \cdot |\mathcal{T}|)$. In line 32, the size of $cand_V$ is bounded by $d_{max}(u)$. Therefore, Lines 32-36 take $\mathcal{O}(d_{max}(u))$ time in the worst case. 
Overall, the time complexity of computing the valid candidate set $C_V^*$ is $\mathcal{O}(d_{max}(u) \cdot d_{max}(v) \cdot |\mathcal{T}|)$.
\end{proof}

According to Theorems~\ref{tc:fc} and \ref{theo:timecomplex-final}, \vfree can reduce each valid candidate set computation cost in \filterv by a factor of $\mathcal{O}(|V|)$.

\section{Experiment}
\label{exp}

\myparagraph{Algorithms}
Note that, \bkalg cannot finish in a reasonable time if directly applied. 
Thus, we equip all the algorithms with the graph filtering technique by default. 
In the experiments, the following algorithms are implemented and evaluated.
$i)$ \textbf{\bkalgp}: \bkalg proposed in Section \ref{sec:method1} with the graph filtering technique; 
$ii)$ \textbf{\filterv}: Algorithm \ref{alg:mfgl} with all the optimizations developed in Section~\ref{sec:method1};
$iii)$ \textbf{\vfree}: Algorithm \ref{alg:final} in Section \ref{sec:method2} with graph filtering technique;
$iv)$ \textbf{\filtervfr}: \filterv without the candidate filtering rule;
$v)$ \textbf{\filtervvm}: \filterv without the verification methods;
$vi)$ \textbf{\filtervr}: \filterv without both the candidate filtering rule and verification strategies; 
$vii)$ \textbf{\vfreer}: \vfree without graph filtering optimization. 

\begin{table}[t]
  \caption{Statistics of datasets}
  \label{tab:dataset}
  \small
  \setlength\tabcolsep{1pt}
  \begin{tabular}{lccccccc}
    \toprule
    \textbf{Dataset}& \textbf{$|U|$}& \textbf{$|V|$}& \textbf{$|\mathcal{E}|$} & $\mathcal{E}~${Type}& \textbf{$|\mathcal{T}|$}& {Scale}& \textbf{($\tau_U,\tau_V,\lambda$)}\\
    \midrule
    D1~(MI)   & 100,000   & 15,648    & 58,951      & diagnose  & 25 & 6month   & (6,2,4)\\ 
    {D2~(Ip)}   & {28,540}   & {37,088}    & {73,153}      & {click}  & {31} & {N/A}   & {(3,2,3)}\\ 
    D3~(diq)  & 25,771    & 1,526     & 133,874     & edit      & 12 & year     & (3,3,3)\\ 
    D4 (vec)  & 33,587    & 2,282     & 339,722     & edit      & 14 & year   & (3,3,3)\\ 
    D5 (LK)   & 337,510   & 42,046    & 605,642     & post      & 35 & year     & (3,3,3)\\ 
    D6 (ben)  & 249,726   & 79,269    & 845,577   & edit      & 17 & year     & (3,3,3)\\ 
    D7 (Wut)  & 530,419   & 175,215   & 2,118,877   & usage     & 39 & month    & (3,2,3)\\ 
    D8 (Bti)  & 767,448   & 204,674   & 2,517,857   & assign    & 22 & year     & (3,3,3)\\ 
    D9 (AR)   & 1,230,916 & 2,146,058 & 5,754,118   & rate      & 21 & year     & (3,3,3)\\ 
    D10 (id)   & 2,183,495 & 125,482   & 7,890,901   & edit      & 59 & quarter  & (3,3,3)\\ 
    D11 (ar)  & 2,943,712 & 209,374   & 13,601,759  & edit      & 57 & quarter  & (3,3,3)\\ 
    D12 (nl)  & 3,800,350 & 220,848   & 28,294,026  & edit      & 65 & quarter  & (10,6,8)\\ 
    D13 (it)  & 4,857,109 & 343,861   & 41,146,957  & edit      & 65 & quarter  & (10,6,8)\\ 
    D14 (fr)  & 8,870,763 & 757,622   & 66,586,964  & edit      & 66 & quarter  & (10,6,8)\\ 
    {D15~(de)}   & {5,910,433}   & {1,025,085}    & {70,745,969}      & {edit}  & {67} & {quarter}   & {(11,11,11)}\\ 
  \bottomrule
\end{tabular}
\end{table}

\myparagraph{Datasets}
{We employ {15} real-world temporal bipartite graphs in our experiments, whose details are shown in Table~\ref{tab:dataset}. $|\mathcal{T}|$ is the number of snapshots. 
D1 (MIMIC-III\footnote{https://physionet.org/content/mimiciii/1.4/}) is a real clinical database that represents relationships between patient and health condition, where the timestamp associated with the edge denotes the time of diagnosis \cite{johnson2016mimic,johnson2016original,johnson2016original3}.
D2 (Ipvevents\footnote{https://tianchi.aliyun.com/dataset/123862}) is a real customer-product network, where edges denote the clicking relationships between customers and products.
Each relationship between the customer and the product is associated with the label to denote whether the customer is a fraudster or not.
The other {13} datasets are obtained from KONECT\footnote{http://konect.cc/networks/}, which are public available.}
To evaluate the impact of time span, we employ a larger dataset D16 (YS), which is a temporal person–song rating network, with $|U|$=624,962, $|V|$=1,000,991 and $|\mathcal{E}|$=256,804,235 from KONECT.

\myparagraph{Parameters and workloads}
We conduct experiments by varying parameters $\tau_U$, $\tau_V$ and $\lambda$, whose default values are shown in the last column of Table~\ref{tab:dataset}.
For each setting, we run each algorithm 10 times and report the average value.
For those experiments that cannot finish within 12 hours, we set them as \textbf{INF}.
All the programs are implemented in standard C++, and performed on a server with an Intel Xeon 2.1GHz CPU and 64 GB main memory.

\begin{figure}[t]
	\centering
	\includegraphics[width=0.95\linewidth]{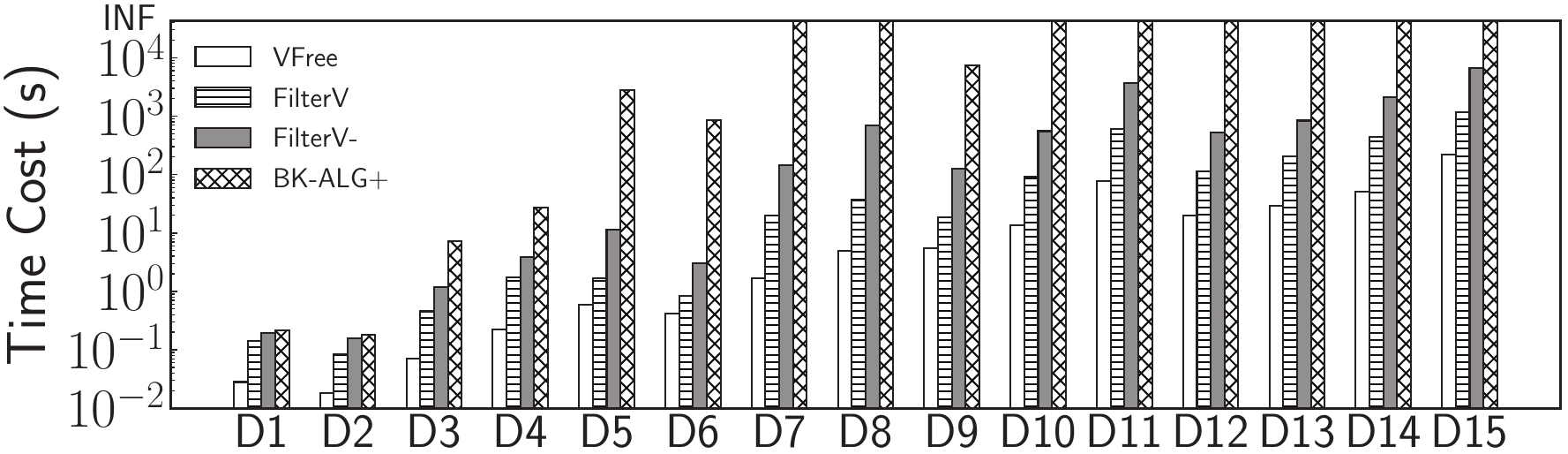}
	
	\caption{Response time on all the datasets}
	\label{fig:exp1}
\end{figure}

\subsection{Efficiency Evaluation}

\myparagraph{Exp-1: Experiments over all the datasets}
Figure \ref{fig:exp1} reports the response time of \bkalgp, \filtervr, \filterv and \vfree over all the datasets with the default settings. 
\filtervr is faster than \bkalgp due to the optimization in the BK framework, i.e., compute the valid candidate set first before the exploration. \filterv further improves \filtervr because of the filtering and verification techniques developed. 
\vfree outperforms the other algorithms with a significant margin. This is because $i)$ \vfree is time-oriented when computing the valid candidate, and the dynamic counting structures can extraordinarily speedup the computation; $ii)$ the search paradigm also significantly reduces the cost of the maximality verification. 
{For instance, on the dataset D14, \bkalgp fails to complete the computation within 12 hours. 
\filtervr and \filterv return the result in 2081 seconds and 445 seconds, respectively. 
\vfree can return the result in 50 seconds.
{On the largest dataset D15 with more than 70 million edges, the response time of \vfree is 218 seconds.}}
In datasets D5, D6 and D9, where \bkalgp can finish in a reasonable time, \vfree can achieve up to three orders of magnitude speedup.



\begin{figure}[t]
    \centering
    \begin{subfigure}{0.15\textwidth}
        \includegraphics[width=\textwidth]{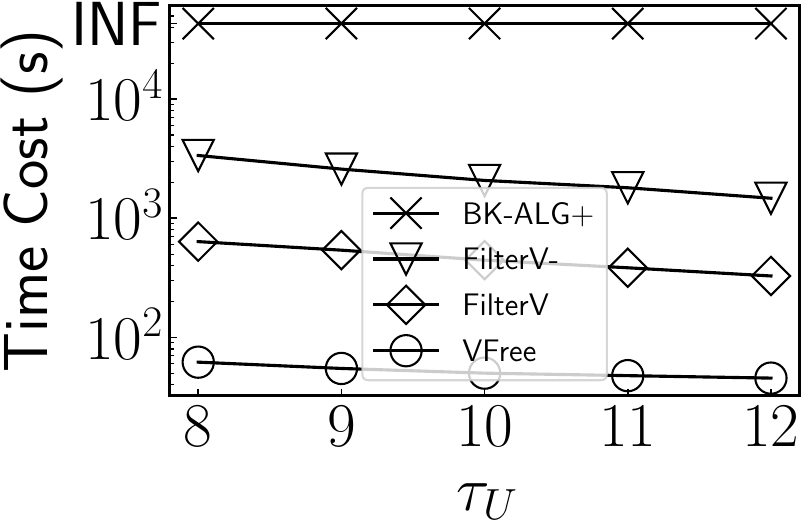}
        \caption{{\footnotesize {D14} ($\tau_U$)}} 
    \end{subfigure}
    \begin{subfigure}{0.15\textwidth}
        \includegraphics[width=\textwidth]{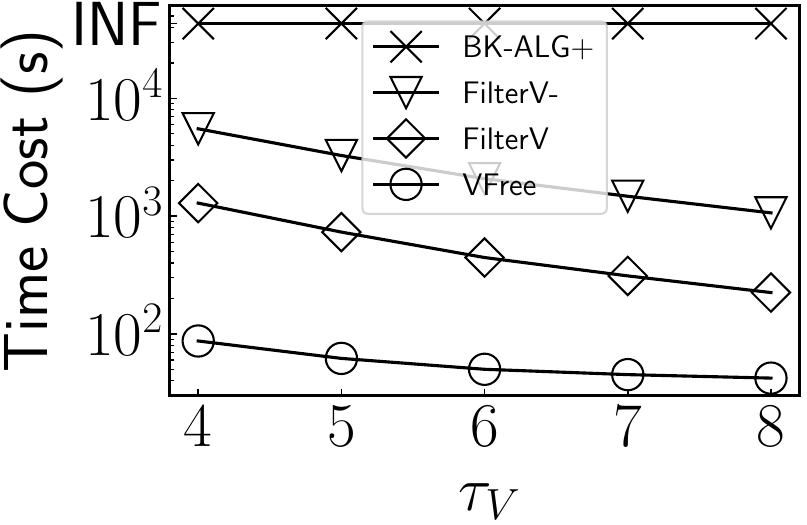}
        \caption{{\footnotesize {D14} ($\tau_V$)}} 
    \end{subfigure}
    \begin{subfigure}{0.15\textwidth}
        \includegraphics[width=\textwidth]{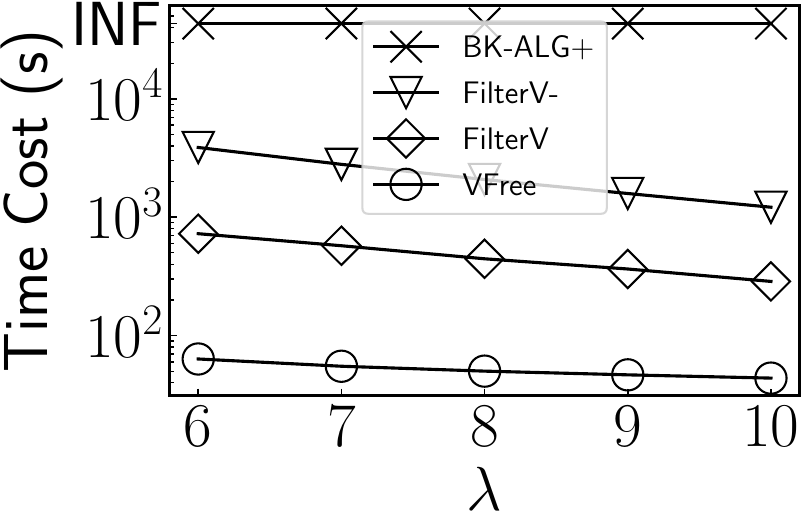}
        \caption{{\footnotesize D14 ($\lambda$)}}
    \end{subfigure}
    \\[0.2cm]
    \begin{subfigure}{0.15\textwidth}
        \includegraphics[width=\textwidth]{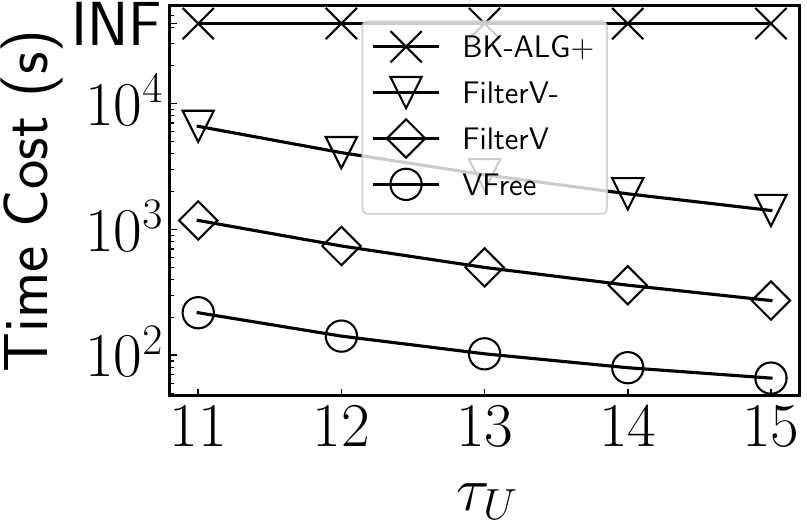}
        \caption{{\footnotesize D15 ($\tau_U$)}} 
    \end{subfigure}
    \begin{subfigure}{0.15\textwidth}
        \includegraphics[width=\textwidth]{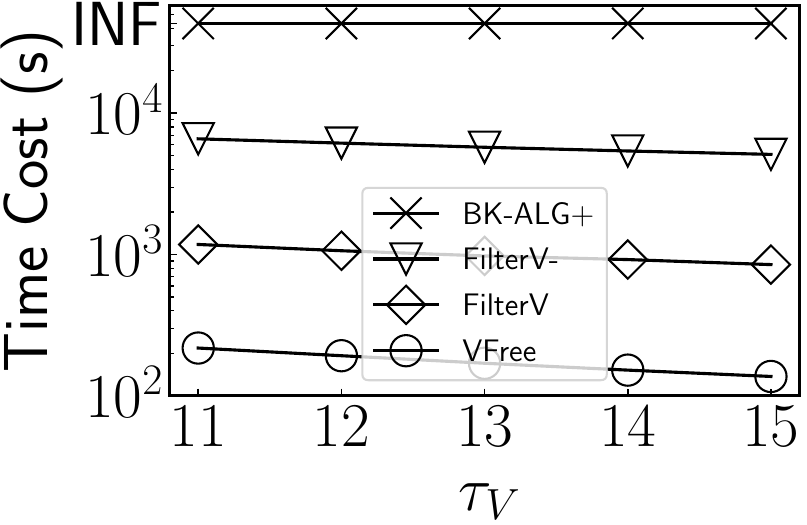}
        \caption{{\footnotesize D15 ($\tau_V$)}} 
    \end{subfigure}
    \begin{subfigure}{0.15\textwidth}
    \includegraphics[width=\textwidth]{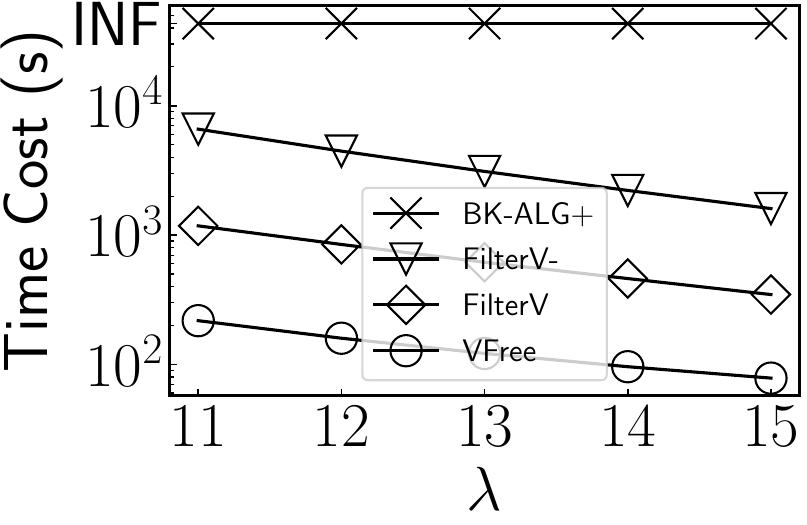}
        \caption{{\footnotesize D15 ($\lambda$)}}
    \end{subfigure}
    \caption{{Response time by varying $\tau_U$, $\tau_V$ and $\lambda$}}
    \label{fig:exp-2}
\end{figure}

\begin{figure}[t]
    \centering
    \begin{subfigure}{0.15\textwidth}
    \includegraphics[width=\textwidth]{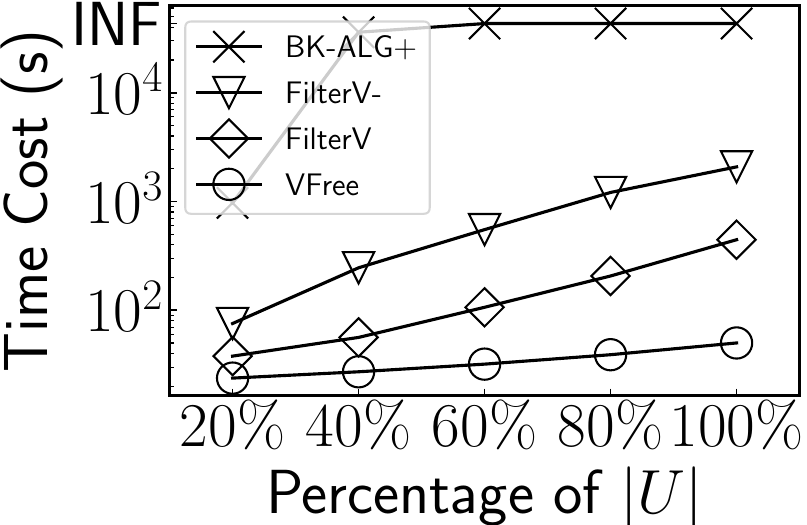}
        \caption{{\footnotesize D14 ($|U|$)}} 
    \end{subfigure}
        \begin{subfigure}{0.15\textwidth}
        \includegraphics[width=\textwidth]{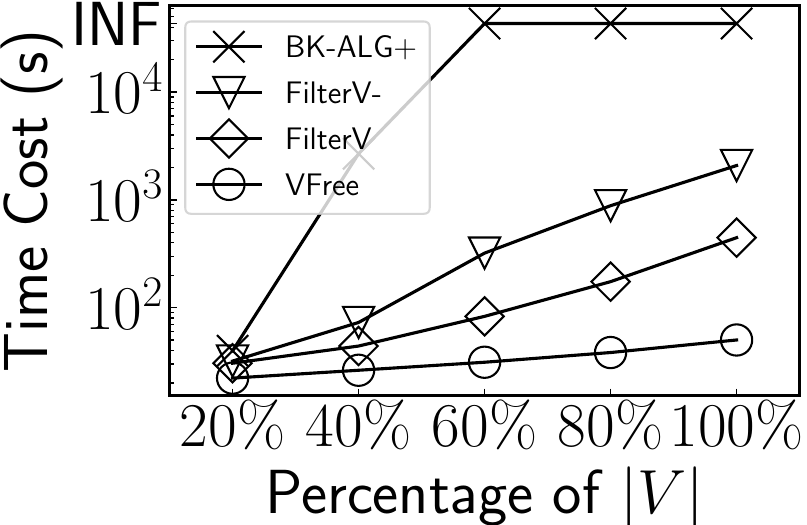}
        \caption{{\footnotesize D14 ($|V|$)}} 
    \end{subfigure}
    \begin{subfigure}{0.15\textwidth}
        \includegraphics[width=\textwidth]{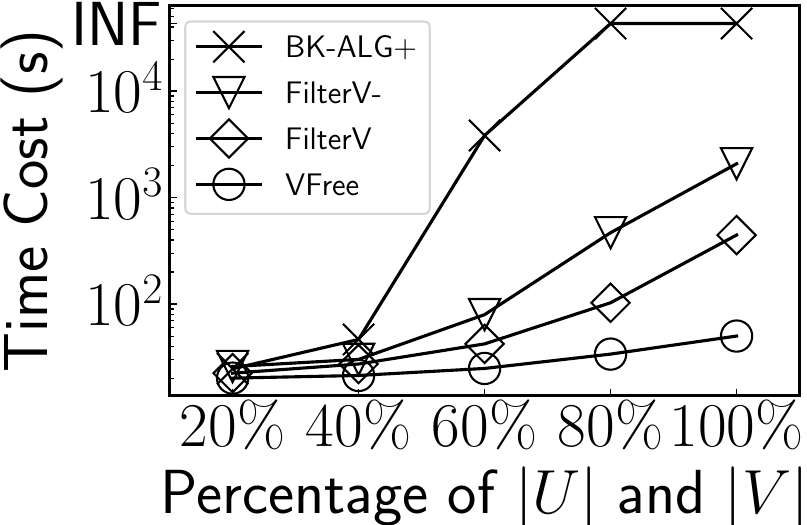}
        \caption{{\footnotesize D14 ($|U|$, $|V|$)}} 
    \end{subfigure}
    \\[0.2cm]
    \begin{subfigure}{0.15\textwidth}
        \includegraphics[width=\textwidth]{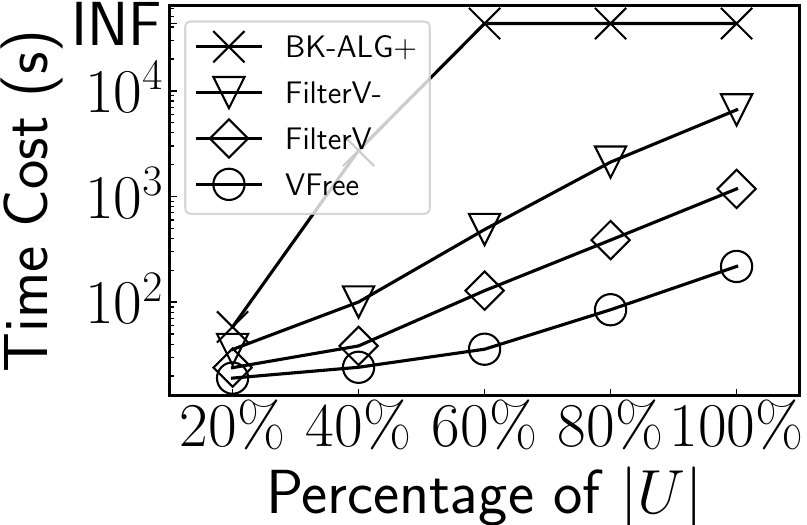}
        \caption{{\footnotesize D15 ($|U|$)}} 
    \end{subfigure}
    \begin{subfigure}{0.15\textwidth}
        \includegraphics[width=\textwidth]{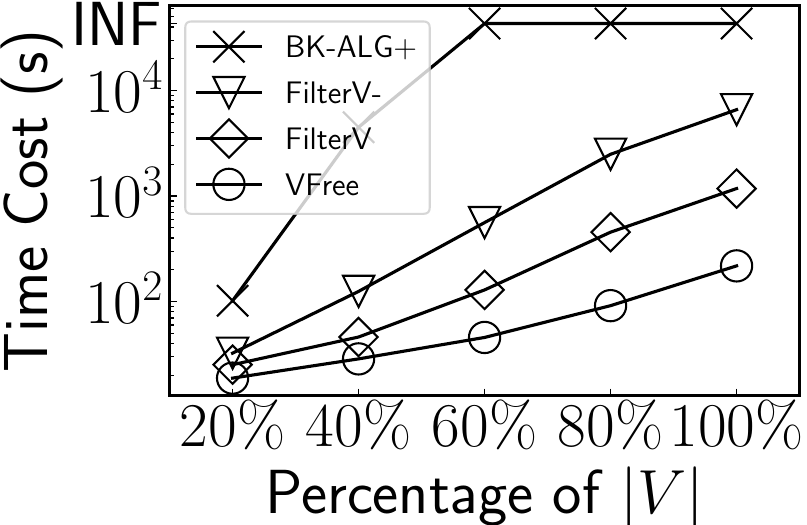}
        \caption{{\footnotesize D15 ($|V|$)}} 
    \end{subfigure}
    \begin{subfigure}{0.15\textwidth}
        \includegraphics[width=\textwidth]{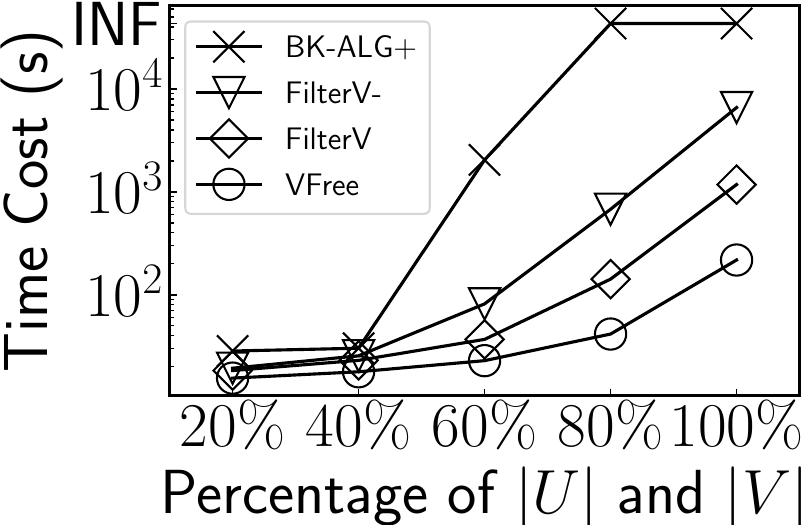}
        \caption{{\footnotesize D15 ($|U|$, $|V|$)}} 
    \end{subfigure}
    \caption{Response time by varying $|U|$ and $|V|$}
    \label{fig:exp-3}
\end{figure}

\myparagraph{Exp-2: Response time by varying parameters}
{In Figures \ref{fig:exp-2}(a) to \ref{fig:exp-2}(f), we report the response time of \bkalgp, \filtervr, \filterv and \vfree on the two largest datasets D14 and D15 by varying parameters $\tau_U$, $\tau_V$ and $\lambda$, respectively. 
Note that, in these experiments, we use default settings for the other unchanged parameters.}
As shown, \vfree is faster than the other three algorithms under all the parameter settings.
Besides, the response time of all the algorithms decreases when the parameters increase. 
{This is because more vertices can be skipped and larger search space can be pruned by the graph filtering technique due to the tighter constraints.}
\bkalgp fails to return the results in a reasonable time on most datasets even with large parameters. 


\myparagraph{Exp-3: Response time by varying $|U|$ and $|V|$}
In this experiment, we use the two largest datasets to demonstrate the scalability of the algorithms.
Specifically, for each dataset, we randomly select 20\%-80\% vertices from $U$ (resp. $V$) to form four new graphs, and Figure \ref{fig:exp-3}(a), \ref{fig:exp-3}(d) (resp. Figure \ref{fig:exp-3}(b), \ref{fig:exp-3}(e)) report the corresponding response time of \bkalgp, \filtervr, \filterv and \vfree. 
Besides, we randomly select 20\%-80\% vertices from both $U$ and $V$ to form four new graphs.
The response time of these four algorithms are illustrated in Figure \ref{fig:exp-3}(c), \ref{fig:exp-3}(f) under the default parameter settings. 
We can find \vfree and \filterv outperform the other algorithms and scale well.
As the number of vertices in $|U|$ and $|V|$ increases, 
the response time of all the algorithms grows due to the larger search space.
As observed, \vfree outperforms the other algorithms by a significant margin and scales well.
For example, when the percentage of $|U|$ is 20\% in Figure \ref{fig:exp-3}(d), the response time of \bkalgp, \filtervr, \filterv and \vfree are 58.38 seconds, 35.74 seconds, 24.05 seconds and 19.15 seconds, respectively.
When the percentage of $|U|$ is 80\%, \bkalgp cannot return results within 12 hours.
The response time of the other three algorithms are 2114.26 seconds, 386.98 seconds and 84.90 seconds, respectively.
In Figure \ref{fig:exp-3}(f), when the percentage of $|U|$ and $|V|$ is 20\% in D15, the response time of \bkalgp, \filtervr, \filterv and \vfree are 28.44 seconds, 19.37 seconds, 18.38 seconds and 15.52 seconds, respectively.
When the percentage of $|U|$ and $|V|$ is 80\%, \bkalgp cannot return results within 12 hours.
The response time of the other three algorithms are 678.02 seconds, 142.06 seconds and 41.58 seconds, respectively.

\begin{figure}[t]
    \centering
    \begin{subfigure}{0.15\textwidth}
    \includegraphics[width=\textwidth]{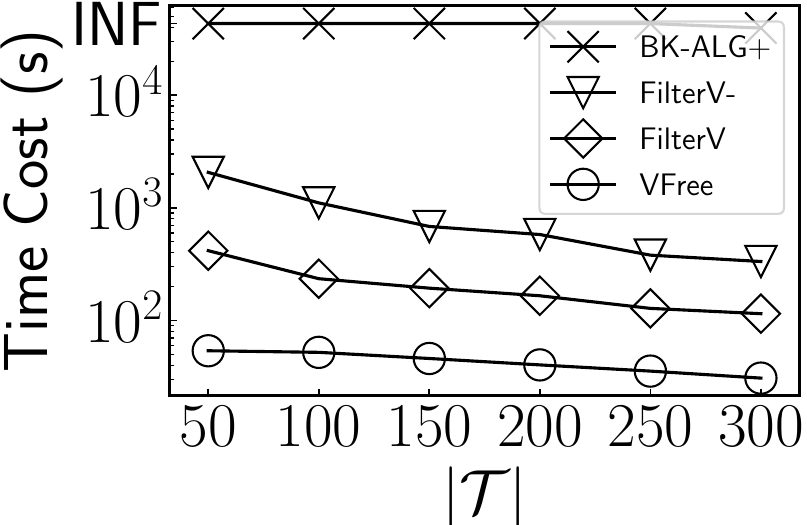}
        \caption{\footnotesize D14 ($|\mathcal{T}|$)} 
    \end{subfigure}
    \begin{subfigure}{0.15\textwidth}
        \includegraphics[width=\textwidth]{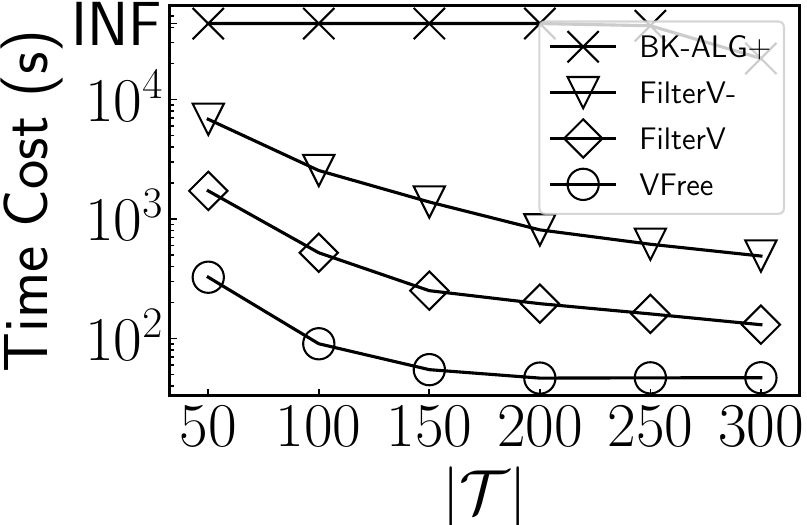}
        \caption{\footnotesize D15 ($|\mathcal{T}|$)} 
    \end{subfigure}
    \begin{subfigure}{0.15\textwidth}
        \includegraphics[width=\textwidth]{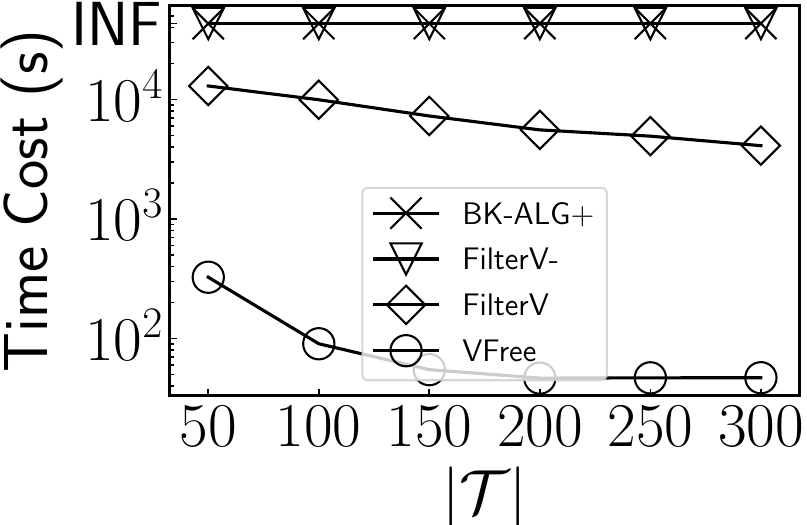}
        \caption{\footnotesize D16 ($|\mathcal{T}|$)} 
    \end{subfigure}
    \\[0.2cm]
    \begin{subfigure}{0.15\textwidth}
    \includegraphics[width=\textwidth]{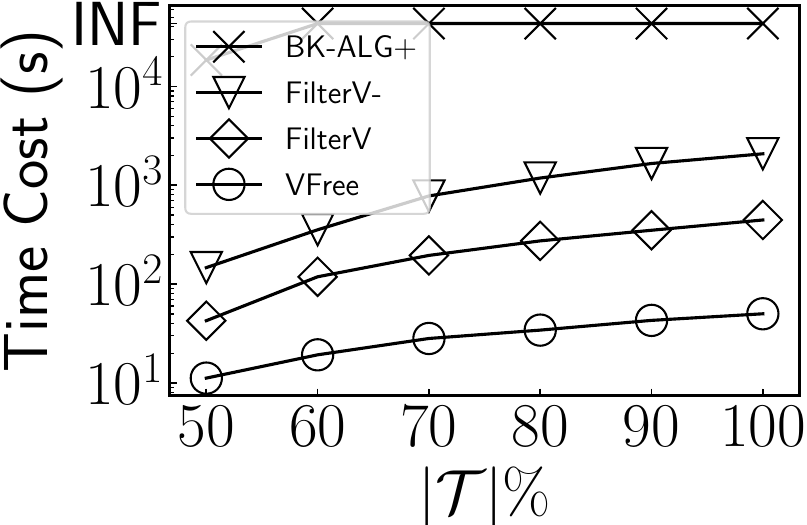}
        \caption{\footnotesize D14 ($|\mathcal{T}|$\%)} 
    \end{subfigure}
    \begin{subfigure}{0.15\textwidth}
        \includegraphics[width=\textwidth]{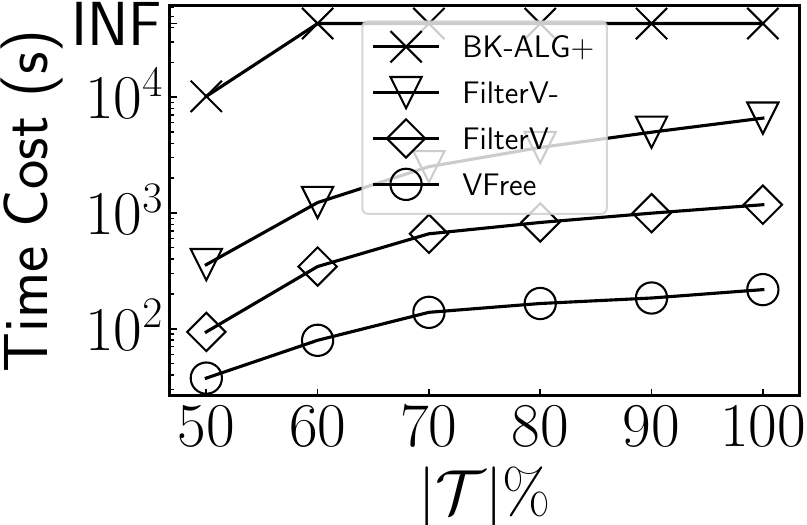}
        \caption{\footnotesize D15 ($|\mathcal{T}|$\%)} 
    \end{subfigure}
    \begin{subfigure}{0.15\textwidth}
        \includegraphics[width=\textwidth]{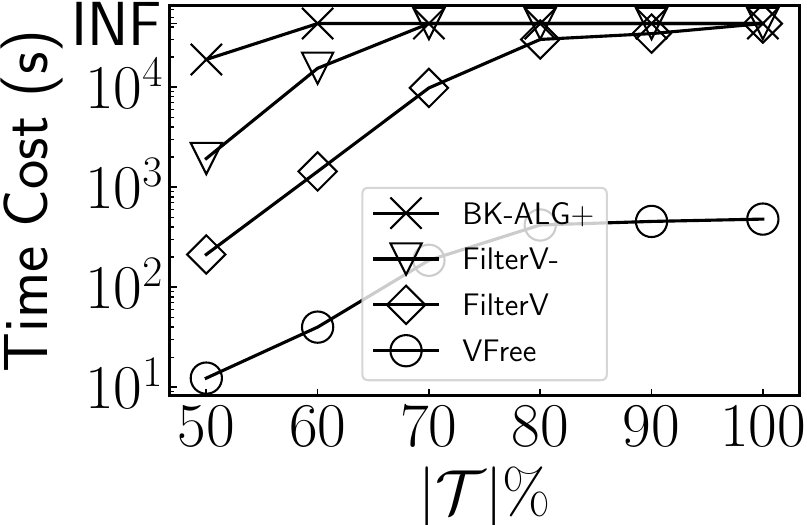}
        \caption{\footnotesize D16 ($|\mathcal{T}|$\%)} 
    \end{subfigure}
    \caption{Response time by varying $|\mathcal{T}|$ and $|\mathcal{T}|\%$}
    \label{fig:exp-varyt}
\end{figure}


\myparagraph{Exp-4: Response time by varying $|\mathcal{T}|$ and $|\mathcal{T}|$\%}
In this experiment, 
    we report the response time of \bkalgp, \filtervr, \filterv and \vfree by varying the settings of time span on D14, D15 and D16. 
    D16
    is a larger graph (with over 256 million edges) employed to further demonstrate the performance of proposed techniques, and default parameters for D16 are $\tau_U=$ 10, $\tau_V=$ 15 and $\lambda=$ 10.
    We conduct two sets of experiments. 
    $i)$ For each dataset, we generate 6 temporal bipartite graphs by varying $|\mathcal{T}|$, i.e.,   
    setting different time scales.
    The results are shown in Figures~\ref{fig:exp-varyt}(a)-\ref{fig:exp-varyt}(c).
    As observed, \vfree outperforms the other algorithms by a significant margin and scales well.
    For example, when $|\mathcal{T}|=300$ in D16,  \bkalgp and \filtervr cannot return results within 12 hours.
    The response time of \filterv and \vfree are 4105.32 seconds and 46.86 seconds, respectively.
    The response time of all the algorithms decreases when $|\mathcal{T}|$ increases.
    This is because, as $|\mathcal{T}|$ increases, the number of edges in each snapshot decreases,
    leading to more space pruned based on the cohesive constraint. 
    $ii)$ For each dataset, we keep $|\mathcal{T}|$ unchanged as the default value, i.e., $66$ for D14, $67$ for D15 and $134$ for D16, and generate 6 temporal bipartite graphs by covering the edges in the first 50\%-100\% timestamps, i.e., $|\mathcal{T}|\%$. The results are shown in Figures~\ref{fig:exp-varyt}(d)-\ref{fig:exp-varyt}(f). 
    As observed, \vfree outperforms the other algorithms.
    The response time of all algorithms grows with the increase of $|\mathcal{T}|\%$ due to the larger search space.
    The above experiments demonstrate that our proposed algorithms are scalable towards different temporal settings.

\begin{figure}[t]
    \centering
    \begin{subfigure}{0.15\textwidth}
        \includegraphics[width=\textwidth]{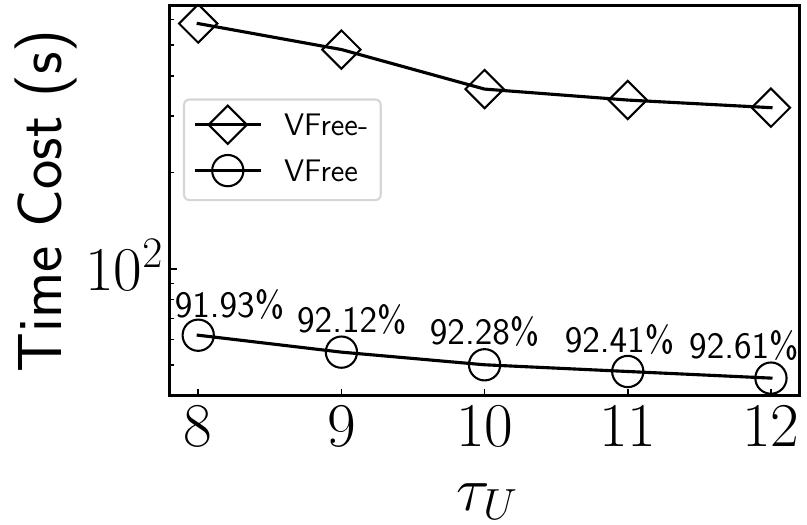}
        \caption{{\footnotesize D14 ($\tau_U$)}} 
    \end{subfigure}
        \begin{subfigure}{0.15\textwidth}
        \includegraphics[width=\textwidth]{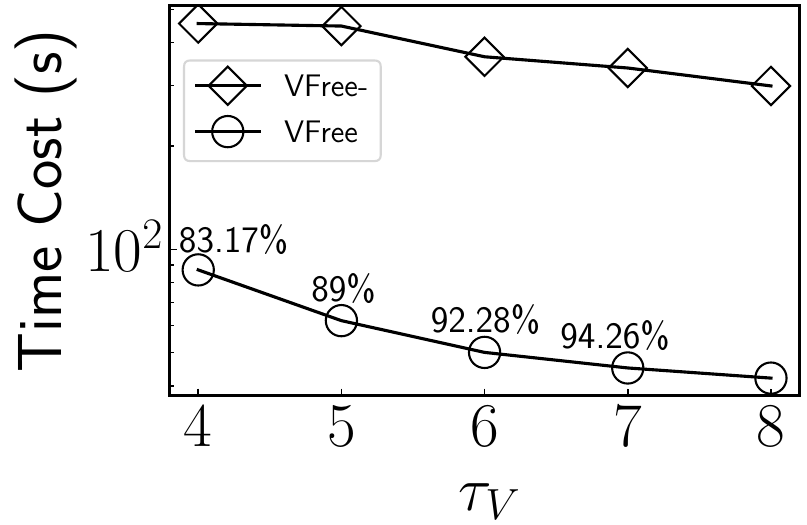}
        \caption{{\footnotesize D14 ($\tau_V$)}} 
    \end{subfigure}
    \begin{subfigure}{0.15\textwidth}
        \includegraphics[width=\textwidth]{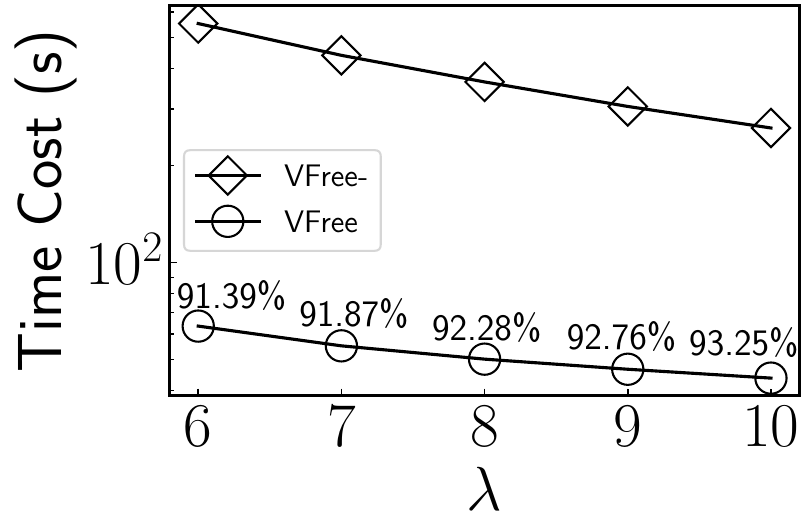}
        \caption{{\footnotesize D14 ($\lambda$)}}
    \end{subfigure}
    \\[0.2cm]
    \begin{subfigure}{0.15\textwidth}
        \includegraphics[width=\textwidth]{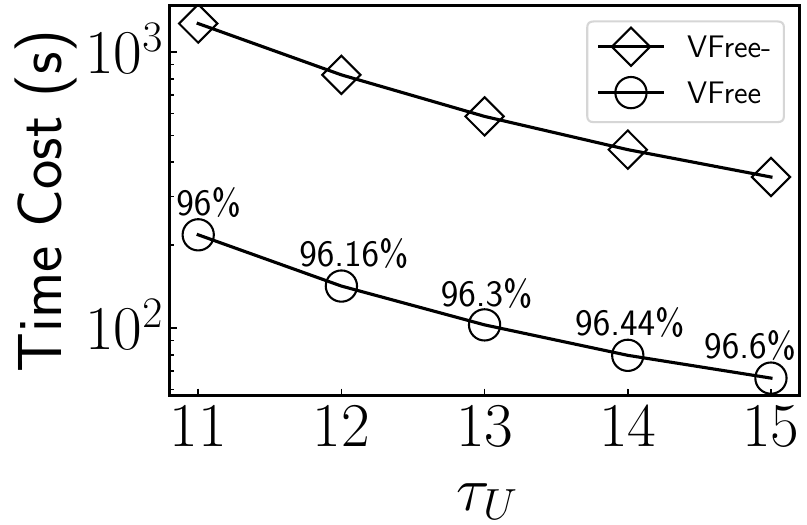}
        \caption{{\footnotesize D15 ($\tau_U$)}} 
    \end{subfigure}
    \begin{subfigure}{0.15\textwidth}
        \includegraphics[width=\textwidth]{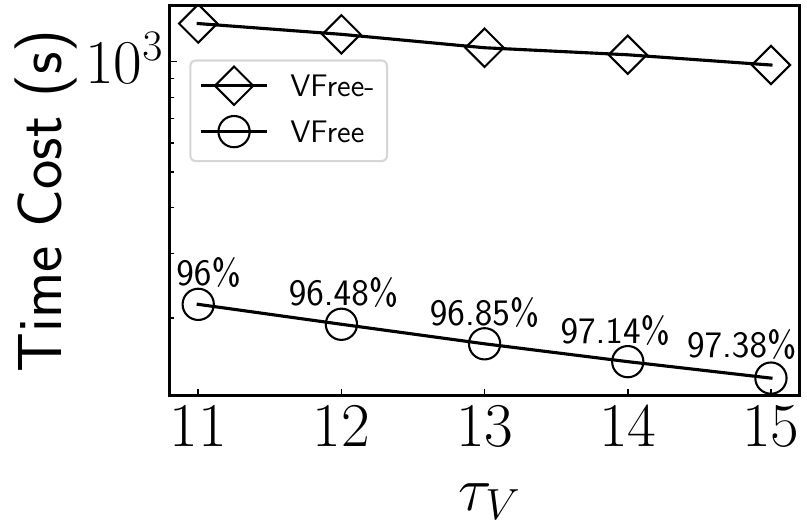}
        \caption{{\footnotesize D15 ($\tau_V$)}} 
    \end{subfigure}
    \begin{subfigure}{0.15\textwidth}
        \includegraphics[width=\textwidth]{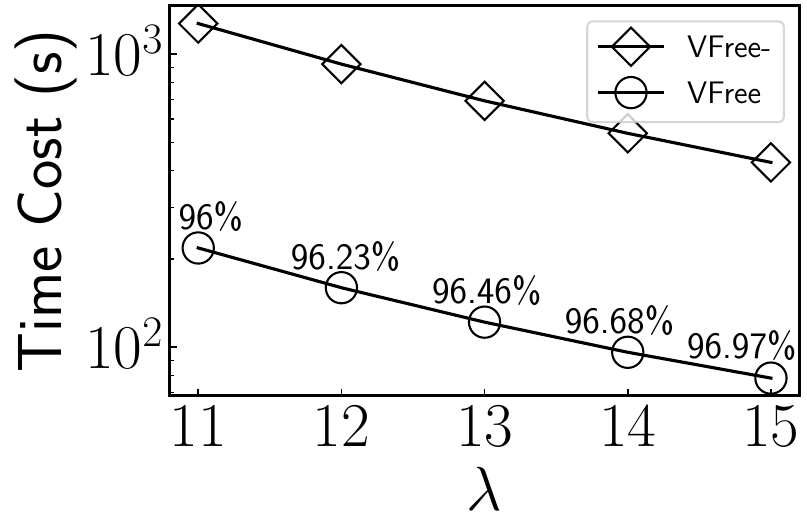}
        \caption{{\footnotesize D15 ($\lambda$)}}
    \end{subfigure}
    \caption{{Evaluation of the graph filtering technique}}
    \label{fig:exp-5-2}
\end{figure}

\begin{figure}[t]
    \centering
    \begin{subfigure}{0.15\textwidth}
        \includegraphics[width=\textwidth]{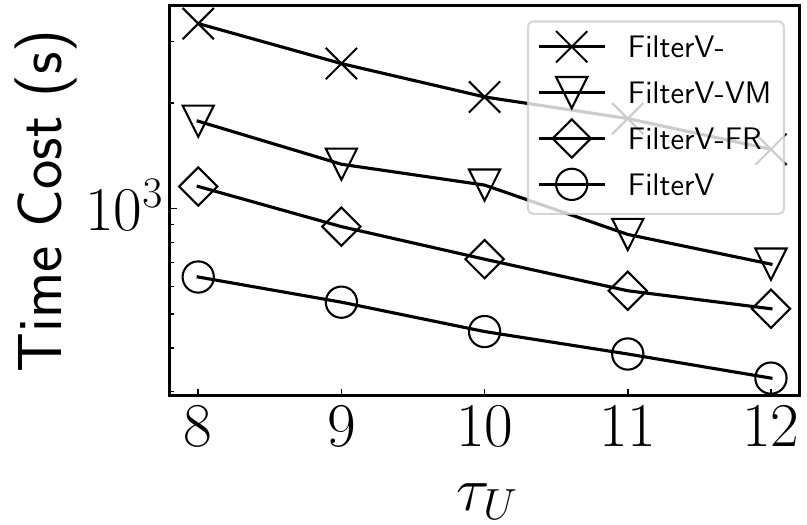}
        \caption{{\footnotesize D14 ($\tau_U$)}} 
    \end{subfigure}
    \begin{subfigure}{0.15\textwidth}
        \includegraphics[width=\textwidth]{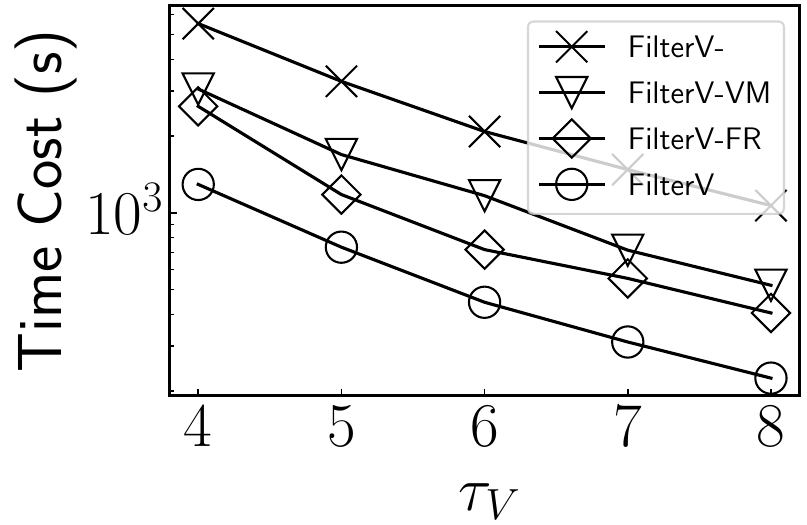}
        \caption{{\footnotesize D14 ($\tau_V$)}} 
    \end{subfigure}
    \begin{subfigure}{0.15\textwidth}
        \includegraphics[width=\textwidth]{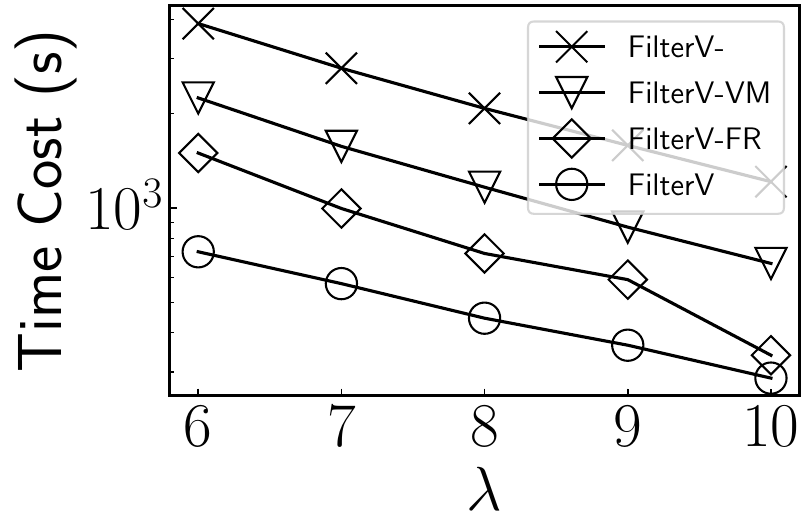}
        \caption{{\footnotesize D14 ($\lambda$)}}
    \end{subfigure}
    \\[0.2cm]
    \begin{subfigure}{0.15\textwidth}
        \includegraphics[width=\textwidth]{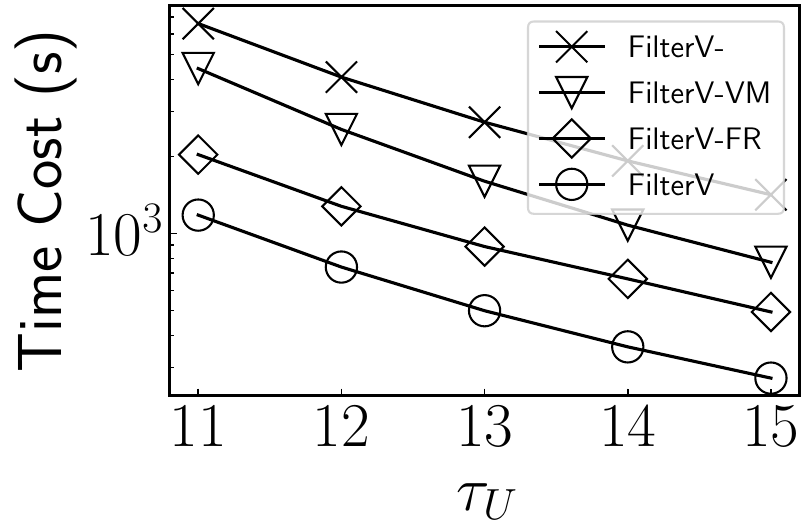}
        \caption{{\footnotesize D15 ($\tau_U$)}} 
    \end{subfigure}
    \begin{subfigure}{0.15\textwidth}
        \includegraphics[width=\textwidth]{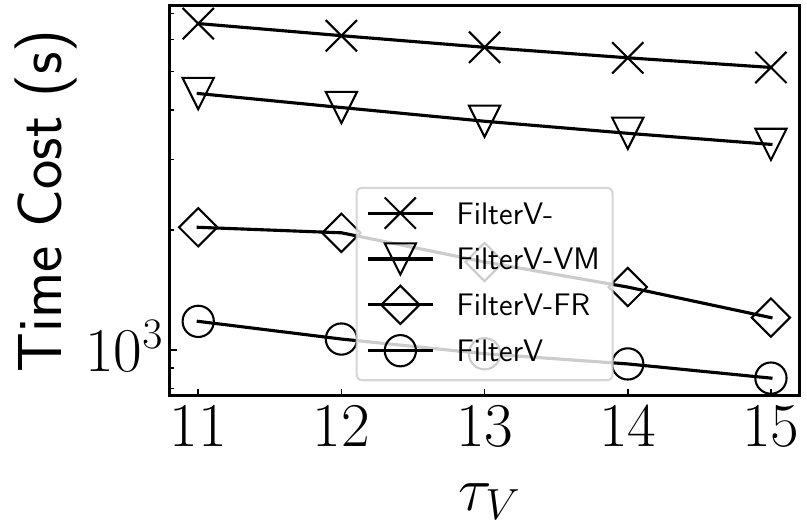}
        \caption{{\footnotesize D15 ($\tau_V$)}} 
    \end{subfigure}
    \begin{subfigure}{0.15\textwidth}
        \includegraphics[width=\textwidth]{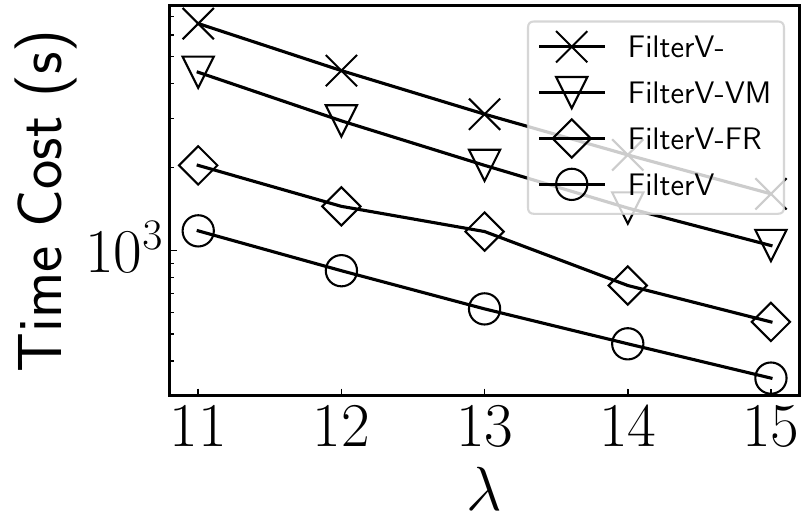}
        \caption{{\footnotesize D15 ($\lambda$)}}
    \end{subfigure}
     
    \caption{{Evaluation of the candidate filtering rule and the verification method}}
    \label{fig:exp-6}
\end{figure}

\myparagraph{Exp-5: Evaluation of graph filtering technique}
{To evaluate the performance of the graph filtering technique, i.e., $(\tau_V, \tau_U, \lambda)$-core based pruning, we conduct the experiments on the two largest datasets D14 and D15 by varying $\tau_U$, $\tau_V$ and $\lambda$, respectively. 
For each experiment, settings for other unchanged parameters follow the default values.
The results are shown in Figure~\ref{fig:exp-5-2}, where the two lines denote the response time of \vfreer and \vfree, and the value above the line of \vfree is the corresponding percentage of edges pruned. 
As shown, compared to the original graph, the graph filtering technique can significantly reduce the number of edges, which validates the effectiveness of the proposed $(\tau_V, \tau_U, \lambda)$-core model. 
{For instance, in Figure~\ref{fig:exp-5-2}(f), 96.97\% edges can be pruned when $\lambda=15$.} 
In most cases, the graph filtering technique can prune more than 90\% of the edges in the graph.
Furthermore, with the increase of parameters, more edges can be pruned due to the tighter constraint.
In addition, the algorithms run faster with the increase of parameters.
Compared with \vfreer, \vfree can achieve {up to 9x} speedup due to the graph filtering technique.
{For example, when $\tau_U=8$ in Figure~\ref{fig:exp-5-2}(a), \vfreer needs 583 seconds to return the result, but \vfree only needs 62 seconds.}}

\myparagraph{Exp-6: Evaluation of the candidate filtering rule and the verification method}
In this experiment, we evaluate the performance of the candidate filtering rule (\ie Lemma \ref{lemma-a2-filter}) and the verification method (\ie Algorithm \ref{alg:fc}).
{In Figure \ref{fig:exp-6}, We report the response time of \filterv, \filtervfr, \filtervvm and \filtervr on D14 and D15  by varying $\tau_U$, $\tau_V$ and $\lambda$, respectively.
The results demonstrate that our filtering rule and validation method can significantly improve the performance of algorithm under all the parameter settings.
For example, in D14 with parameters (10,6,8), 
{\filtervfr, \filtervvm and \filtervr take 716 seconds,
1166 seconds and 2081 seconds to return the result, respectively.}
\filterv can return the result within 445 seconds, which verifies the advantage of proposed techniques.}

\begin{figure}[t]
    \centering
    \begin{subfigure}{0.15\textwidth}
        \includegraphics[width=\textwidth]{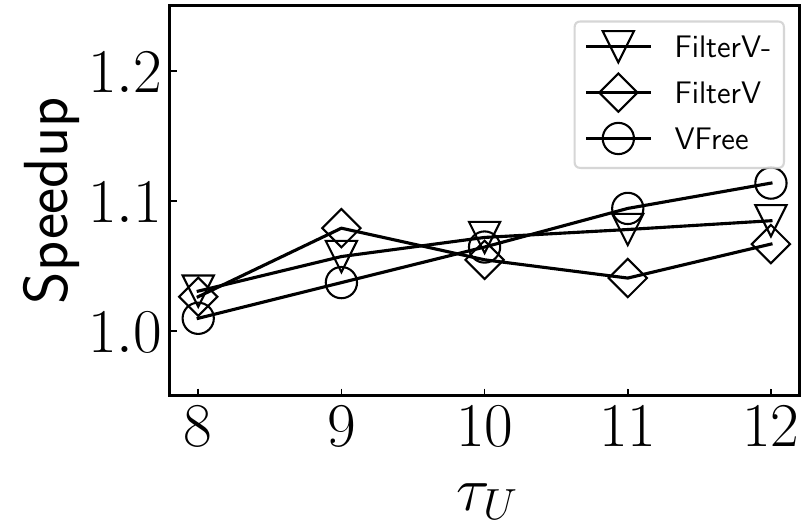}
        \caption{\footnotesize D14 ($\tau_U$)} 
    \end{subfigure}
    \begin{subfigure}{0.15\textwidth}
        \includegraphics[width=\textwidth]{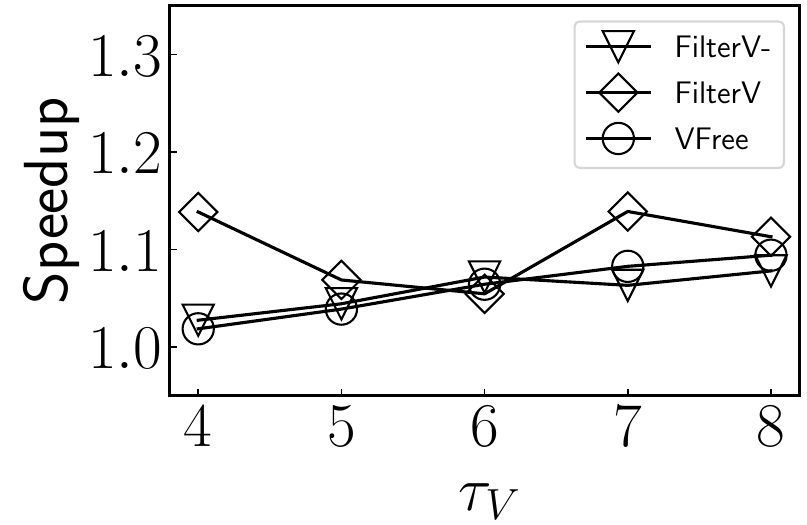}
        \caption{\footnotesize D14 ($\tau_V$)} 
    \end{subfigure}
    \begin{subfigure}{0.15\textwidth}
        \includegraphics[width=\textwidth]{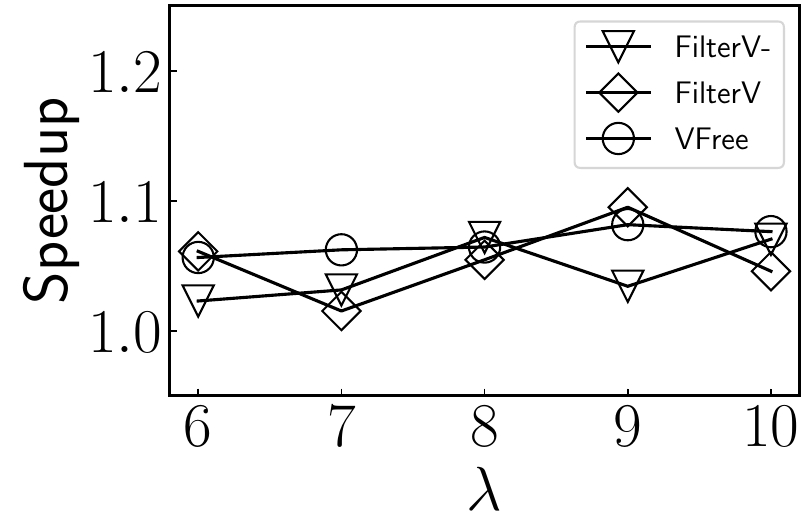}
        \caption{\footnotesize D14 ($\lambda$)}
    \end{subfigure}
    \\[0.2cm]
    \begin{subfigure}{0.15\textwidth}
        \includegraphics[width=\textwidth]{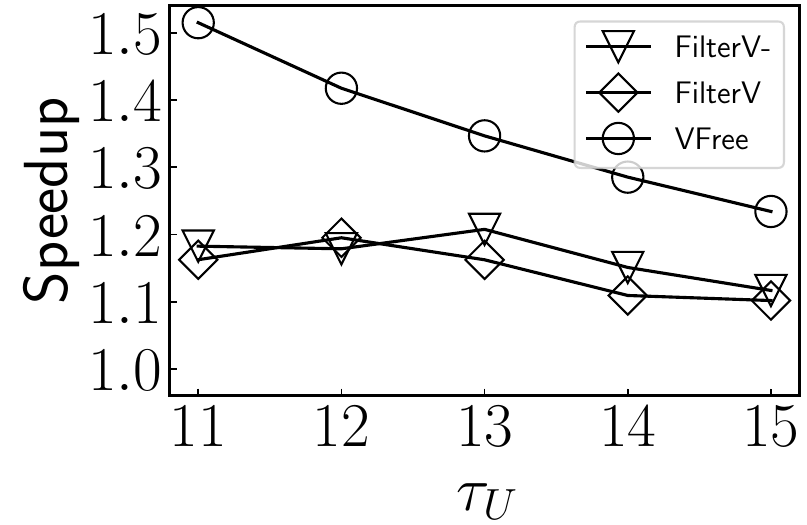}
        \caption{\footnotesize D15 ($\tau_U$)} 
    \end{subfigure}
    \begin{subfigure}{0.15\textwidth}
        \includegraphics[width=\textwidth]{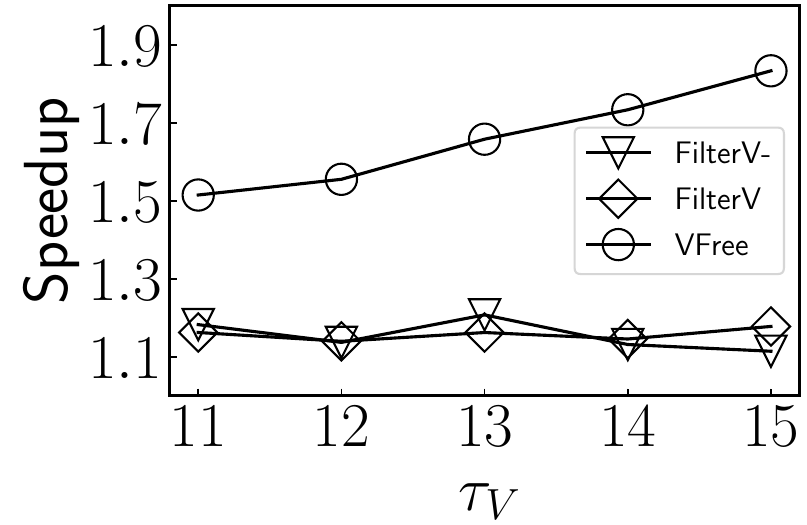}
        \caption{\footnotesize D15 ($\tau_V$)} 
    \end{subfigure}
    \begin{subfigure}{0.15\textwidth}
        \includegraphics[width=\textwidth]{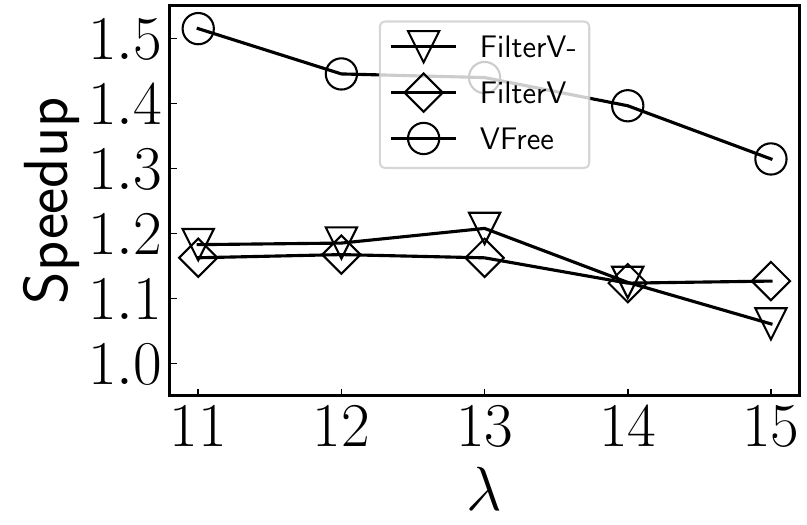}
        \caption{\footnotesize D15 ($\lambda$)}
    \end{subfigure}
     
    \caption{Evaluation of the vertex ID reorder technique}
    \label{fig:exp-new7}
\end{figure}

\myparagraph{Exp-7: Evaluation of the ID reorder technique}
To evaluate the impact of the ID reorder technique in \vfree, we conduct the experiments on D14 and D15 by varying $\tau_U$, $\tau_V$ and $\lambda$, respectively.
For \filtervr, \filterv and \vfree, we report the \textit{speedup ratio} of ID reorder technique, which is calculated by $\frac{\text{response time of algorithms without ID reorder}}{\text{response time of algorithms with ID reorder}}$.
The results are shown in Figure~\ref{fig:exp-new7}.
Note that, we omit \bkalgp here, since it cannot finish on D14 and D15 within 12 hours
even with the ID reorder technique.
As observed, the ID reorder technique can improve the efficiency of algorithms. 
For instance, the speedup for \vfree is up to 1.8x.

\subsection{Effectiveness Evaluation}

\exppara{Exp-8: Multimorbidity detection}
To demonstrate the effectiveness of proposed model, we conduct a case study on D1 (MIMIC-III) for potential multimorbidity detection. 
    MIMIC-III is a real clinical dataset representing relationships between patients and health conditions, where the timestamp denotes the time of diagnosis \cite{johnson2016mimic,johnson2016original,johnson2016original3}.
    By applying our model \mrg, we can model the situation where multiple health conditions appear for different patients at multiple times simultaneously.
    The partially identified \mrgs are shown in Table~\ref{tab:case1}. 
    For instance, `SEPSIS' and `PNEUMONIA' are highly correlated, since pneumonia is a common cause of sepsis.
    Moreover, 
    we compare \mrg with two variants based on existing models,~i.e., $i)$~maximal frequent $(\tau_U,\tau_V)$-biclique (\textsf{MFB}) and $ii)$ maximal static group (\textsf{MSG}).
    \textsf{MFB} is the maximal $(\tau_U,\tau_V)$-biclique with frequency at least $\lambda$, i.e., appearing in at least $\lambda$ snapshots.
    \textsf{MSG} is the maximal group included in $(\tau_U,\tau_V)$-biclique of the corresponding static graph.
  As shown in Table~\ref{tab:case1}, these two models cannot obtain practical results.
    Specifically, we cannot find any results by applying \textsf{MFB} model due to the tight constraint.
    For \textsf{MSG}, the identified groups may be too large to provide practical information due to neglect of temporal aspect.

\begin{table}
\caption{Case study on D1 ($\tau_U=\tau_V=2$ and $\lambda=6$)}
\label{tab:case1}
\centering
  \footnotesize
\begin{tabularx}{\linewidth}{|m{0.7cm}<{\centering}|m{7.03cm}|}
    \hline
    \centering \textbf{Model} & \multicolumn{1}{c|}{\textbf{Partial results}} \\
    \hline
        \hline
    \centering \mrg & $\{$SEPSIS, PNEUMONIA$\}$, $\{$GASTROINTESTINAL BLEED, LOWER GI BLEED$\}$, $\{$ASTHMA, COPD EXACERBATION, PNEUMONIA$\}$, $\{$UPPER GI BLEED, LOWER GI BLEED$\}$, $~\dots$ \\
    \hline
    \centering \textsf{MSG} &  $\{$CHRONIC OBST PULM DISEASE, CHRONIC~OBSTRUCTIVE PULMONARY, RESPIRATORY FAILURE, PNEUMONIA, COPD EXACERBATION, ASTHMA$\}$, 
    $\{$HYPERTENSIVE EMERGENCY, HYPERTENSIVE URGENCY, ABDOMINAL PAIN, DIABETIC KETOACIDOSIS$\}$, $~\dots$ \\
    \hline
    \centering \textsf{MFB} & {N/A}\\
    \hline
\end{tabularx}
\end{table}

\begin{figure}[t]
    \centering
    \begin{subfigure}{0.15\textwidth}
    \includegraphics[width=\textwidth]{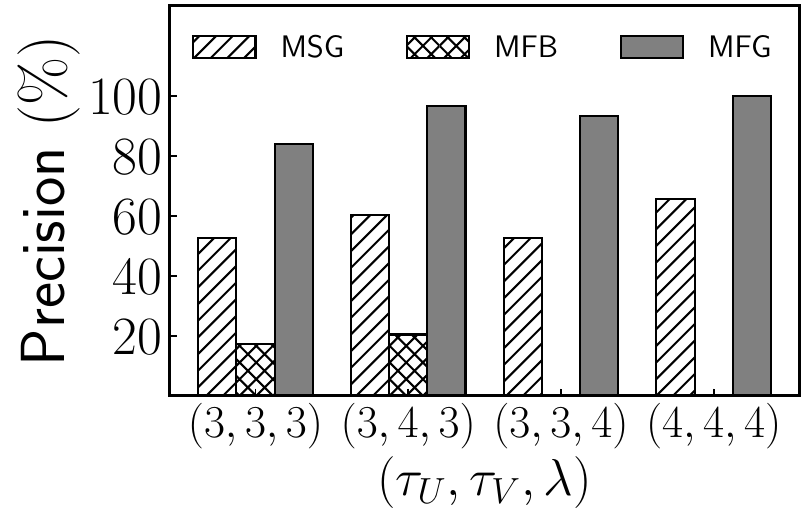}
        \caption{\small Precision} 
    \end{subfigure}
        \begin{subfigure}{0.15\textwidth}
        \includegraphics[width=\textwidth]{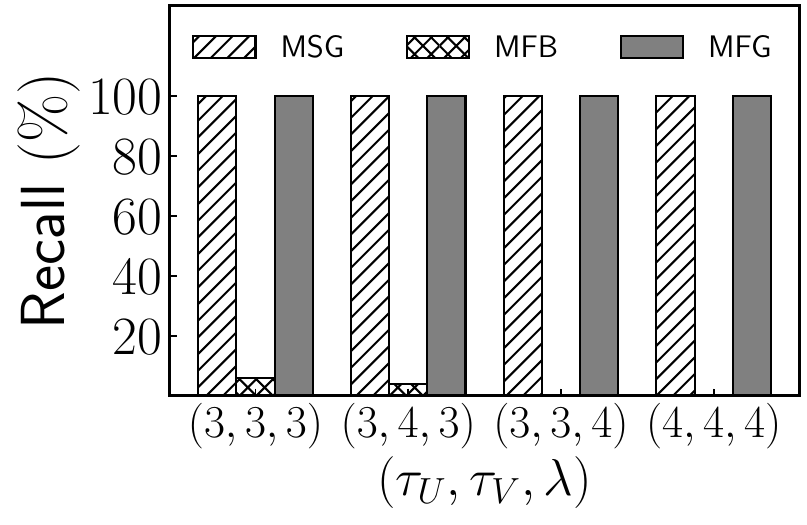}
        \caption{\small Recall} 
    \end{subfigure}
    \begin{subfigure}{0.15\textwidth}
        \includegraphics[width=\textwidth]{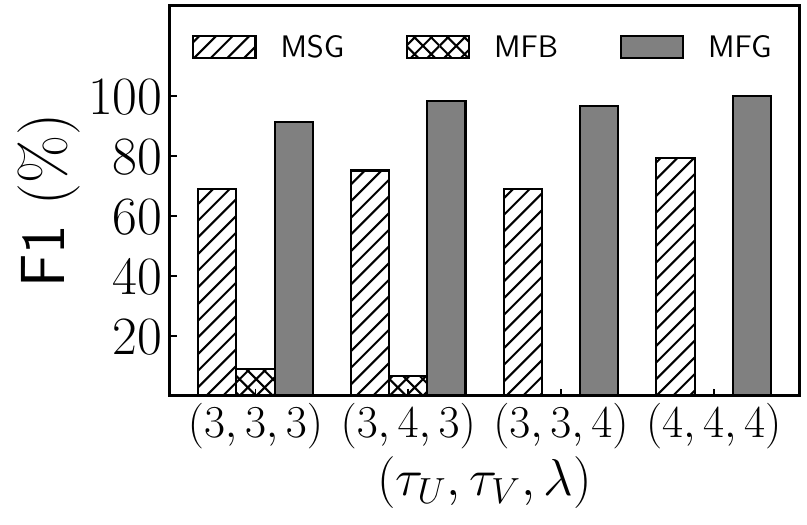}
        \caption{\small F1 score} 
    \end{subfigure}
    \caption{Case study on D9}
    \label{fig:exp-fd}
\end{figure}

\begin{figure}[t]
    \centering
    \begin{subfigure}{0.15\textwidth}
    \includegraphics[width=\textwidth]{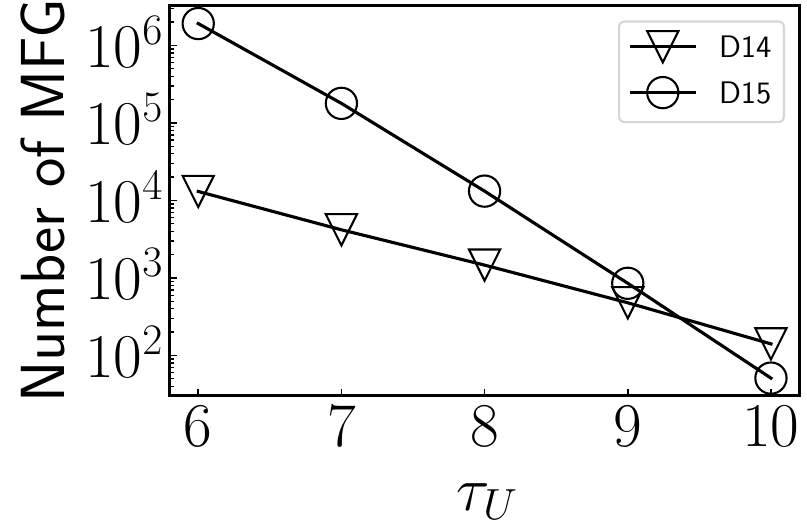}
        \caption{\footnotesize D14, D15 ($\tau_U$)} 
    \end{subfigure}
        \begin{subfigure}{0.15\textwidth}
        \includegraphics[width=\textwidth]{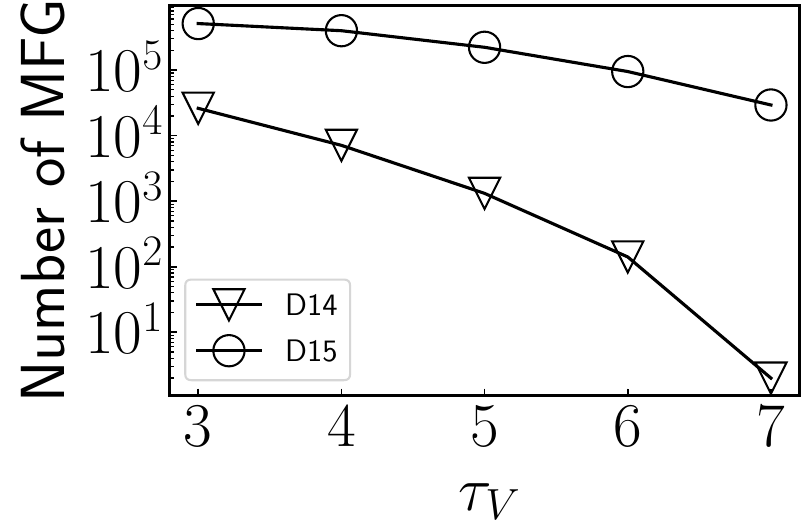}
        \caption{\footnotesize D14, D15 ($\tau_V$)} 
    \end{subfigure}
    \begin{subfigure}{0.15\textwidth}
        \includegraphics[width=\textwidth]{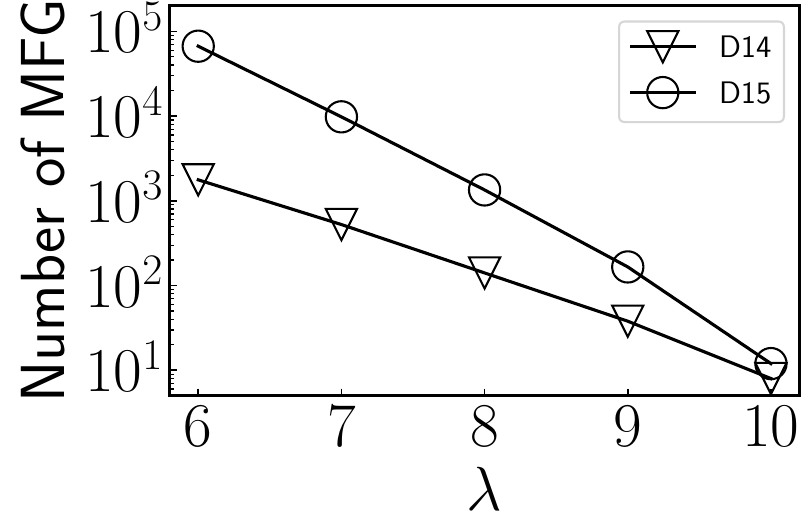}
        \caption{\footnotesize D14, D15 ($\lambda$)} 
    \end{subfigure}
    \caption{Number of \mrgs by varying parameters}
    \label{fig:exp-number}
\end{figure}


\myparagraph{Exp-9: Fraud detection}
In this case study, we demonstrate the application for a fraud detection task in face of the camouflage attack~\cite{hooi2016fraudar} on D9 (Amazon Ratings) compared with \textsf{MFB} and \textsf{MSG}.
Considering a practical scenario with a random camouflage attack on D9, we randomly select five timestamps and choose 2K products at each selected timestamp. 
    Specifically, at each of the selected timestamps, we introduce a fraudulent block containing 1K fake users, 1.5K fake comments and 0.5K camouflage comments to the data.
    The same fake users are used across selected timestamps.
Besides, at each selected timestamp, fake comments (resp. camouflage comments) are generated randomly between these injected fake users and the chosen products (resp. real products).
We categorize all users identified within the detected subgraphs as fake items and all others as real one, and report the precision, recall and F1 score in Figure~\ref{fig:exp-fd}.
    As observed, \mrg can achieve better results compared to the others, which demonstrates the effectiveness of our proposed model. 
    For \textsf{MFB}, it cannot find results under some settings, and we omit the corresponding value in the figure.
    For \textsf{MSG}, it usually has high recall but low precision since it
    would find large groups involving many users.

\myparagraph{Exp-10: Number of \mrgs}
In this experiment, we report the number of returned \mrgs on D14 and D15 by varying different parameters. 
The results are shown in Figure~\ref{fig:exp-number}. 
The settings for other unchanged parameters follow the default values. 
The number of \mrgs decreases with the increase of three parameters due to the tighter constraint required.
As shown, even with large parameters, we can still find \mrg patterns, which demonstrates the applicability of our model in different settings for real scenarios.
\section{Related Work}
\label{sec:rel}

In the literature, many subgraph models have been proposed for temporal unipartite graphs (\eg~\cite{ma2017fast,qin2019mining,qin2022mining,li2023persistent}).
Ma \etal~\cite{ma2017fast} study the connected temporal subgraph whose aggregated graph has the maximum cohesive density on temporal weighted graphs.
They develop algorithms based on the assumption that the weights of all edges
are increasing or decreasing in the same direction.
Li \etal~\cite{li2018persistent} define a novel $k$-core based model to capture the persistence of a community.
They first propose a novel temporal graph reduction method and then develop a novel branch and bound algorithm with several powerful pruning rules to find the result.
\cite{qin2019mining} leverages the clique model and investigates the $\sigma$-periodic $k$-clique enumeration problem.
The authors first prune the input temporal graph based on two novel relaxed periodic clique models.
Then, they propose a graph transformation
technique and efficiently enumerate the results on the transformed graph.
Qin~\etal~\cite{qin2020mining} propose a novel concept named $(\mu,\tau,\epsilon)$-stable core, to characterize the stable core nodes of the clusters.
They call a node $u$ is $(\mu,\tau,\epsilon)$-stable core if it has no less than $\mu$ neighbors that are simultaneously similar to itself in at least $\tau$ snapshots of the temporal graph.
Using the similar ideas of \cite{qin2019mining}, they first propose the weak and strong cores to significantly prune the unpromising nodes, and then identify all the stable clusters from the remaining graph.
There are also some studies that focus on bursting subgraph mining, 
which are established on the time-interval based constraint.
In \cite{qin2022mining}, 
Qin \etal study the problem of mining bursting cores, where the $(l,\delta)$-maximal bursting core model is developed.
They propose a novel dynamic programming algorithm to speedup the calculation of the segment density.
Zhang \etal \cite{zhang2023discovering} propose a frequency bursting pattern in temporal graphs. It tries to model the interactive behaviors that accumulate their frequencies the fastest.
Chu \etal \cite{chu2019online} aim to find the top-$k$ density bursting subgraphs, each of which is the subgraph that accumulates their density at the fastest speed in the temporal graph.
In \cite{lou2021time}, a new metric named $\mathbb{T}$-cohesiveness is proposed by jointly considering the time and topology dimensions.
For temporal bipartite graphs, 
Li \etal \cite{li2023persistent} study the community search problem, and a persistent community model is proposed based on $(\alpha,\beta)$-core, which has different semantics from ours. 
If we consider a temporal graph as a set of snapshots, the problems of frequent subgraph mining and multi-layer cohesive subgraph mining are correlated. 
Specifically, the goal of frequent subgraph mining is identifying frequently appeared subgraphs from a collection of graphs (e.g.,~\cite{jiang2013survey,liu2022satmargin}). 
As for multi-layer graph mining, many models are proposed to pinpoint structures, such as the $d$-coherent core and firm truss~\cite{zhu2018diversified,hashemi2022firmcore,behrouz2022firmtruss}.
Besides the above studies, there are some works focusing on dynamic graphs~(e.g., \cite{DBLP:journals/pvldb/TangLHGXL22,ditursi2017local,liu2019finding}).
For instance, Tang~\etal~\cite{DBLP:journals/pvldb/TangLHGXL22} propose a novel $(\theta,k)$-core reliable community in the weighted dynamic networks.
That is, a $k$-core with each edge weight no less than the weight threshold $\theta$ spans over a period of time.
To find the most reliable local community with the highest reliability score, they first filter the unpromising edges from the graph.
Then they develop an index structure and devise an index-based dynamic programming search algorithm.
Although many studies have been conducted over temporal or dynamic graphs, it is non-trivial to extend the existing solutions to our problem.

As the most cohesive structure in bipartite graphs, the biclique model has attracted significant attention.
The problem of maximal biclique enumeration has been widely studied in static bipartite graphs~\cite{abidi2020pivot,chen2022efficient,zhang2014finding,alexe2004consensus,liu2006efficient}.
Most of the existing biclique enumeration algorithms conduct the search by expanding the vertices from one side. 
Then, we can intersect their neighbors to form the corresponding biclique
\cite{abidi2020pivot,chen2022efficient,zhang2014finding}.
In \cite{gely2009enumeration}, the authors reduce the maximal biclique enumeration problem to the maximal clique enumeration problem.
In \cite{zhang2014finding}, Zhang \etal remove the unpromising candidates from the search branches by employing the branch-and-bound framework.
In \cite{abidi2020pivot}, Abidi \etal develop a novel pivot-based technique to block non-maximal search branches. 
In~\cite{chen2022efficient}, the authors accelerate the process of maximal biclique enumeration by proposing the concepts of unilateral coreness for individual vertices, unilateral order for each vertex set and unilateral convergence for a large sparse bipartite graph. 
Generally, in the unilateral core, for all vertices on one side of it, the number of their two-hop neighbors within must be no less than $k$.
The concepts of unilateral coreness, order and convergence are proposed based on the concept of unilateral core.
However, these concepts are orthogonal to our model and have different semantics.
In~\cite{yao2022identifying}, the authors propose a dense bipartite subgraph model, which considers similarity between vertices from the same side on static bipartite graphs.
To the best of our knowledge, although there are a few works considering the one-layer properties in bipartite graphs, no existing works study the unilateral frequency model.
There are also some studies that consider the biclique problem in various bipartite graphs, such as signed bipartite graph~\cite{sun2022maximal,sun2023maximum}.  
As we can see, few studies consider the case of temporal bipartite graphs. 
Moreover, due to the unilateral property of our model, the existing studies cannot be extended to solve our problem efficiently.

\section{Conclusion}
\label{sec:conc}

Temporal bipartite graph is an important data structure to model many real-world applications. 
To analyze the properties of temporal bipartite graphs, 
in this paper, we propose a novel model, named maximal $\lambda$-\frequent group, by considering both unilateral cohesiveness and temporal aspect. 
We first introduce a filter-and-verification method by extending the BK framework. 
Novel filtering techniques and array-based checking method are developed.
To further improve the performance, a verification-free approach is proposed based on advanced dynamic counting strategy, which can significantly reduce the cost of valid candidate set computation and avoid explicit maximality verification. Experiments over {15 real-world datasets} confirm the efficiency and effectiveness of proposed techniques and model.


\clearpage

\bibliographystyle{ACM-Reference-Format}
\bibliography{ref}

\end{document}